\documentclass[11pt]{article}

\usepackage{amsmath}
\usepackage{amssymb}
\usepackage{amsthm}
\usepackage{color}
\usepackage{bm}
\usepackage{times}
\usepackage{hyperref}
\usepackage{here}
\usepackage{url}
\usepackage[linesnumbered,ruled,vlined]{algorithm2e}

\usepackage{authblk}
\usepackage{fullpage}
\usepackage{times}

\newtheorem{thm}{Theorem}

\newtheorem{prop}[thm]{Proposition}
\newtheorem{lem}[thm]{Lemma}
\newtheorem{cla}[thm]{Claim}

\newtheorem{df}[thm]{Definition}
\newtheorem{cor}[thm]{Corollary}

\newcounter{num}
\newcommand{\ZZ}{\mathbb{Z}} % 整数環
\newcommand{\RR}{\mathbb{R}} % 実数体
\newcommand{\ca}[1]{\mathcal{#1}} % mathcal

\newcommand{\ol}[1]{\overline{#1}} % オーバーライン
\newcommand{\wt}[1]{\widetilde{#1}} % widetilde
\newcommand{\vol}{\mathrm{vol}}
\newcommand{\sweep}{\mathrm{sweep}}
\newcommand{\1}{\mathbf{1}}
\newcommand{\0}{\mathbf{0}}

\newcommand{\argmin}{\mathop{\mathrm{argmin}}}
\newcommand{\argmax}{\mathop{\mathrm{argmax}}}

\def\inprod<#1>{\left\langle #1 \right\rangle} % innerproduct
\title{Finding Cheeger Cuts in Hypergraphs via Heat Equation} %TODO Please add

\author[1,2]{Masahiro Ikeda\thanks{masahiro.ikeda@riken.jp}}
\author[1]{Atsushi Miyauchi\thanks{atsushi.miyauchi.hv@riken.jp}}
\author[1,2]{Yuuki Takai\thanks{yuuki.takai@riken.jp}}
\author[3]{Yuichi Yoshida\thanks{yyoshida@nii.ac.jp}}
\affil[1]{RIKEN Center for Advanced Intelligence Project}
\affil[2]{Keio University}
\affil[3]{National Institute of Informatics}
\begin{document}

\maketitle

%TODO mandatory: add short abstract of the document
\begin{abstract}
  Cheeger's inequality states that a tightly connected subset can be extracted from a graph $G$ using an eigenvector of the normalized Laplacian associated with $G$.
  More specifically, we can compute a subset with conductance $O(\sqrt{\phi_G})$, where
  %$r$ is the maximum size of hyperedges and
  $\phi_G$ is the minimum conductance of a set in $G$.

  It has recently been shown that Cheeger's inequality can be extended to hypergraphs.
  However, as the normalized Laplacian of a hypergraph is no longer a matrix, we can only approximate its eigenvectors; this causes a loss in the conductance of the obtained subset.
  To address this problem, we here consider the heat equation on hypergraphs, which is a differential equation exploiting the normalized Laplacian.
  We show that the heat equation has a unique global solution and that we can extract a subset with conductance $\sqrt{\phi_G}$ from the solution under a mild condition.
  An analogous result also holds for directed graphs.
\end{abstract}

\thispagestyle{empty}
\setcounter{page}{0}
\newpage

%!TEX root=main.tex

\section{Introduction}

The goal of spectral clustering of graphs is to extract tightly connected communities from a given weighted graph $G=(V,E,w)$, where $w\colon E \to \mathbb{R}_+$ is a weight function, using eigenvectors of matrices associated with $G$.
One of the most fundamental results in this area is Cheeger's inequality, which relates the second-smallest eigenvalue of the normalized Laplacian of $G$ and the conductance of $G$.
Here, the \emph{(random-walk) normalized Laplacian} of $G$ is defined as $\mathcal{L}_G=I-A_G D_G^{-1}$, where $A_G \in \mathbb{R}^{V \times V}$ and $D_G \in \mathbb{R}^{V \times V}$ are the (weighted) adjacency matrix and the (weighted) degree matrix, respectively, of $G$, that is, $D_G$ is a diagonal matrix with the $(v,v)$-th element for $v \in V$ being the (weighted) degree $d_G(v) := \sum_{e \in E \mid v \in e}w(e)$ of $v$.
% is the degree $d(v)$ of $v \in V$.
% diagonal matrix with the $(v,v)$-th element equal to the degree $d(v)$ of $v \in V$.
Note that all eigenvalues of $\mathcal{L}_G$ are non-negative and the smallest eigenvalue is always zero, as $\mathcal{L}_G(D_G\bm{1})=\bm{0}$, where $\bm{1}$ is the all-one vector and $\bm{0}$ is the zero vector.
The \emph{conductance} of a set $\emptyset \subsetneq S \subsetneq V$ is defined as
\[
  \phi_G(S) := \frac{\sum_{e \in \partial_G(S)}w(e)}{\min\{\vol_G(S),\vol_G(V \setminus S) \}},
\]
where $\partial_G(S)$ is the set of edges between $S$ and $V \setminus S$, and $\vol_G(S) := \sum_{v \in S}d_G(v)$ is the \emph{volume} of $S$.
Intuitively, smaller $\phi_G(S)$ corresponds to more tightly connected $S$.
The \emph{conductance} of $G$ is the minimum conductance of a set in $G$; that is, $\phi_G := \min_{\emptyset \subsetneq S\subsetneq V}\phi_G(S)$.
Then, Cheeger's inequality~\cite{Alon:1986gz,Alon:1985jg} states that
\begin{align}
  \frac{\lambda_G}{2} \leq \phi_G \leq \sqrt{2\lambda_G}, \label{eq:cheeger-graph}
\end{align}
where $\lambda_G \in \mathbb{R}_+$ is the second-smallest eigenvalue of $\mathcal{L}_G$.
The second inequality of~\eqref{eq:cheeger-graph} is algorithmic in the sense that we can compute a set $\emptyset \subsetneq S \subsetneq V$ with conductance of at most $\sqrt{2\lambda_G} = O(\sqrt{\phi_G})$, which is called a \emph{Cheeger cut}, in polynomial time from an eigenvector corresponding to $\lambda_G$.
Moreover, Cheeger's inequality is tight in the sense that computing a set with conductance $o(\sqrt{\phi_G})$ is NP-hard~\cite{Raghavendra:2012fc}, assuming the small set expansion hypothesis (SSEH)~\cite{Raghavendra:2010jg}.

Several attempts to extend Cheeger's inequality to hypergraphs have been made.
To explain the known results, we first extend the concepts of conductance and the normalized Laplacian to hypergraphs.
Let $G=(V,E,w)$ be a weighted hypergraph, where $w\colon E \to \mathbb{R}_+$ is a weight function.
The \emph{(weighted) degree} of a vertex $v \in V$ is $d_G(v) := \sum_{e \in E\mid v \in e}w(e)$.
For a vertex set $\emptyset \subsetneq S \subsetneq V$, the \emph{conductance} of $S$ is defined as
\[
  \phi_G(S) := \frac{\sum_{e \in \partial_G(S)}w(e) }{\min\{\vol_G(S),\vol_G(V \setminus S) \}},
\]
where $\partial_G(S)$ is the set of hyperedges intersecting both $S$ and $V \setminus S$, and $\vol_G(S)$ has the same definition as usual graph.
The \emph{conductance} of $G$ is defined as $\phi_G := \min_{\emptyset \subsetneq S\subsetneq V}\phi_G(S)$.

The \emph{normalized Laplacian} $\mathcal{L}_G\colon \mathbb{R}^V \to 2^{\mathbb{R}^V}$\footnote{We note that the range of the Laplacian defined in~\cite{Chan:2018eu} is $\mathbb{R}^V$ instead of $2^{\mathbb{R}^V}$. They chose the value of $\mathcal{L}_G(\bm{x})$ so that it satisfies a necessary condition that the heat equation~\eqref{eq:heat-equation-hypergraph} has a solution, which makes it unique. Hence as long as we consider the solution to~\eqref{eq:heat-equation-hypergraph}, our Laplacian behaves as that defined in~\cite{Chan:2018eu}. Nevertheless, we keep the range $2^{\mathbb{R}^V}$ for a general treatment of the heat equation using the theory of monotone operators. See Section~\ref{sec:solution} for more details.} of a hypergraph $G$~\cite{Chan:2018eu,Yoshida:2017zz} is multi-valued and no longer linear (see Section~\ref{sec:pre} for a detailed definition).
In the simplest setting that the hypergraph $G$ is unweighted and $d$-regular, that is, every vertex has degree $d$, and the elements of the given vector $\bm{x} \in \mathbb{R}^V$ are pairwise distinct, the $\mathcal{L}_G$ acts as follows: We create an undirected graph $G_{\bm{x}}$ on $V$ from $G$ by adding for each hyperedge $e \in E$ an undirected edge $uv$, where $u = \argmin_{w \in e}\bm{x}(w)$ and $v = \argmax_{w \in e}\bm{x}(w)$, then return $\mathcal{L}_{G_{\bm{x}}}\bm{x}$.
% More explicitly, $L(\bm{x})$ can be written as follows:
% \[
%   L(\bm{x})(v) = \frac{1}{d}\left(\sum_{e \in E: e^{\min}_{\bm{\rho}} \neq v \neq e^{\max}_{\bm{\rho}}=v} \bm{\rho}(v)-\sum_{e \in E: e^{\min}_{\bm{\rho}}=v}\bm{\rho}(e^{\max}_{\bm{\rho}}) - \sum_{e \in E: e^{\max}_{\bm{\rho}}=v}\bm{\rho}(e^{\min}_{\bm{\rho}}) \right).
% \]

When $\mathcal{L}_G(\bm{v}) \ni \lambda \bm{v}$ holds for $\lambda \in \mathbb{R}$ and $\bm{v} \neq \bm{0}$, we can state that $\lambda$ and $\bm{v}$ are an \emph{eigenvalue} and an \emph{eigenvector}, respectively, of $\mathcal{L}_G$.
As with the graph case, all eigenvalues of $\mathcal{L}_G$ are non-negative and the first eigenvalue is zero as $\mathcal{L}_G(D\bm{1})=\bm{0}$ holds.
Moreover, the second-smallest eigenvalue $\lambda_G \in \mathbb{R}_+$ exists.
Cheeger's inequality for hypergraphs~\cite{Chan:2018eu,Yoshida:2017zz} states that
\begin{align}
  \frac{\lambda_G}{2} \leq \phi_G \leq 2\sqrt{\lambda_G}. \label{eq:cheeger-hyper}
\end{align}
Again, the second inequality is algorithmic: If we can compute an eigenvector corresponding to $\lambda_G$, we can obtain a Cheeger cut; that is, a set $\emptyset \subsetneq S \subsetneq V$ with $\phi_G(S) = O(\sqrt{\phi_G})$, in polynomial time.
Unlike the undirected graph case, however, only an $O(\log n)$-approximation algorithm is available for computing $\lambda_G$~\cite{Yoshida:2017zz}. Further, this approximation ratio is tight under the SSEH~\cite{Chan:2018eu}.
Hence, the following natural question arises: Can we compute a Cheeger cut without computing $\lambda_G$ and applying Cheeger's inequality on the corresponding eigenvector?
% (at least when $\lambda_G = O(1/\log^2 n)$ so as not to contradict to SSEH).

To answer this question, we consider the following differential equation called the \emph{heat equation}~\cite{Chan:2018eu}:
\begin{align}
  \frac{d \bm{\rho}_t}{dt} \in - \mathcal{L}_G(\bm{\rho}_t) \quad \text{and} \quad \bm{\rho}_0 = \bm{s},
  \tag{$\mathrm{HE};\bm{s}$}
  \label{eq:heat-equation-hypergraph}
\end{align}
where $\bm{s} \in \mathbb{R}^V$ is an initial vector.
% In the case that the hypergraph is $d$-regular and elements of $\bm{\rho} \in \mathbb{R}^V$ are pairwise distinct, we have
% \[
%   \frac{d\bm{\rho}}{dt}(v) = \frac{1}{d}\left(\sum_{e \in E: e^{\min}_{\bm{\rho}}=v}\bm{\rho}(e^{\max}_{\bm{\rho}}) + \sum_{e \in E: e^{\max}_{\bm{\rho}}=v}\bm{\rho}(e^{\min}_{\bm{\rho}}) - \sum_{e \in E: e^{\min}_{\bm{\rho}} \neq v \neq e^{\max}_{\bm{\rho}}=v} \bm{\rho}(v)\right).
% \]
Intuitively, we gradually diffuse values (or \emph{heat}) on vertices along hyperedges so that the maximum and minimum values in each hyperedge become closer.
We can show that~\eqref{eq:heat-equation-hypergraph} always has a unique global solution for $t \geq 0$\footnote{Previous works~\cite{Chan:2018eu,Yoshida:2017zz} only guaranteed that it has a local solution for $0 \leq t \leq T_0$ for some $T_0> 0$.} using the theory of monotone operators and evolution equations~\cite{komura1967nonlinear},~\cite{miyadera1992nonlinear} (see Section~\ref{sec:solution} for details), and let $\bm{\rho}_t^{\bm{s}} \in \mathbb{R}^V$ be the solution at time $t \geq 0$.
In particular, $\bm{\rho}_0^{\bm{s}} = \bm{s}$ holds.
In addition, if $\sum_{v \in V}\bm{s}(v) = 1$, we can show that $\sum_{v \in V}\bm{\rho}_t^{\bm{s}}(v)=1$ holds for any $t \geq 0$, and that $\bm{\rho}_{t}^{\bm{s}}$ converges to $\bm{\pi} \in \mathbb{R}^V$ as $t\rightarrow\infty$ when $G$ is connected, where $\bm{\pi}(v) := d_G(v)/\vol(V)$ (see~\cite[Theorem 3.4]{Chan:2018eu}). Throughout this paper, we assume that
hypergraph $G$ is connected.

For a vector $\bm{x} \in \mathbb{R}^V$, let $\sweep(\bm{x})$ denote the set of all \emph{sweep sets} with respect to $\bm{x}$; that is, sets of the form either $\{v \in V \mid \bm{x}(v) \geq \tau \}$ or $\{v \in V \mid \bm{x}(v) \leq \tau \}$, for some $\tau \in \mathbb{R}$.
We want to show that the conductance of the sweep set of a vector obtained from the heat equation is small.
To this end, for $T\geq 0$, we introduce a key quantity in our analysis:
\[
g_v(T) = -\left.\frac{d}{dt} \log \| \bm{\rho}_t^{\bm{\pi}_v} - \bm{\pi} \|_{D^{-1}}^2 \right|_{t=T},
% = 2\frac{\langle
% \bm{\rho}_T^{\bm{\pi}_v}, \ca{L}\bm{\rho}_T^{\bm{\pi}_v}\rangle_{D^{-1}}}
% {\| \bm{\rho}_T^{\bm{\pi}_v} - \bm{\pi} \|^2_{D^{-1}}}.
\]
where $\|\bm{x}\|_{D^{-1}}^2 := \bm{x}^\top D^{-1}\bm{x}$, which quantifies how fast the heat converges to the limit, that is, $\bm{\pi}$.
We can show that $g_v(T)$ is twice the Rayleigh quotient of
$D^{-1/2}(\bm{\rho}^{\bm{\pi}_v}_T -\bm{\pi})$ with respect to
 the \emph{normalized Laplacian} $\bm{x} \mapsto D^{-1/2}\ca{L}_G(D^{-1/2}(\bm{x}))$.
%   and
%  is close to twice of an eigenvalue of
% the normalized Laplacian $\ca{L}_{G_{v,T}}$ of the graph $G_{v,T} = G_{\bm{\rho}^{\bm{\pi}_v}_T}$, which is determined by $G$ and $\bm{\rho}^{\bm{\pi}_v}_T$ as above.
This fact, combined with Cheeger's inequality for hypergraphs, implies that $g_v(T)$ captures the minimum conductance of a sweep set obtained from $\bm{\rho}^{\bm{\pi}_v}_T$.
% We give an upper bound on the minimum conductance of a sweep set in terms of $g_v(T)$.
\begin{thm}\label{thm:main-intro}
  Let $G=(V,E,w)$ be a weighted hypergraph and $\emptyset \subsetneq S \subsetneq V$ be a set.
  For any $t > 0$, we have
  \[
    g_v(T) \geq {(\kappa_{T,t}^v)}^2
  \]
  where $\kappa_{T,t}^v := \min \{ \phi_G(S') \mid \xi \in [T,t], S' \in \sweep(\bm{\rho}^{\bm{\pi}_v}_\xi) \}$ and $\bm{\pi}_v \in \mathbb{R}^V$ is a vector for which $\bm{\pi}_v(v) = 1$ and $\bm{\pi}_v(u) = 0$ for $u \neq v$.
\end{thm}
%Let $\emptyset \subsetneq S^* \subsetneq V$ be the set that achieves $\phi(S^*) = \phi$.
Let $\bm{u}_2(G_{v,T})$ be an eigenvector corresponding to the
second smallest eigenvalue $\lambda_2(G_{v,T})$ of the
normalized Laplacian $\ca{L}_{G_{v,T}}$. Then, by using Cheeger's inequality for undirected graphs, we have the following corollary:
\begin{cor}\label{cor:main-intro-cor}
 Assume that $\langle \bm{u}_2(G_{v,T}), \rho_T^{\bm{\pi}_v} \rangle_{D^{-1}} \neq 0$ holds. Then, we have
\[
  4\phi_G h_v(T) \geq {(\kappa_{T,t}^v)}^2,
\]
where $h_v(T) = g_v(T)/\lambda_2(G_{v,T})$.
\end{cor}
We can show that $h_v(T)$ is close to $1$ when $T$ is large.
Hence, Corollary~\ref{cor:main-intro-cor} implies that, when $T$ is sufficiently large, under the assumption in the statement, we can obtain a set $\emptyset \subsetneq S \subsetneq V$ such that $\phi_G(S) = O(\sqrt{\phi_G})$, thereby avoiding the problem of computing the second smallest eigenvalue $\lambda_G$ of the hypergraph normalized Laplacian $\ca{L}_G$.
Algorithm~\ref{alg:main} gives a pseudocode of our algorithm.
\begin{algorithm}[t]
\caption{Finding Cheeger Cuts via Heat Equation}\label{alg:main}
\SetKwInOut{Input}{Input}
\SetKwInOut{Output}{Output}
\Input{ Hypergraph $G=(V,E,w)$ and $t > T > 0$}
\Output{ $S\subseteq V$ }
Select an arbitrary $v\in V$\;
Solve ($\mathrm{HE};\bm{\pi}_v$) to obtain $\bm{\rho}^{\bm{\pi}_v}_\xi$ for $\xi\in [T,t]$\;
$S_\text{out}\leftarrow \argmin_{S\in \bigcup_{\xi \in [T,t]}\sweep(\bm{\rho}^{\bm{\pi}_v}_\xi)}\phi_G(S)$\;
\Return{\emph{$S_\text{out}$}}
%\For{each $v\in V$}{
%Solve ($\mathrm{HE};\bm{\pi}_v$) to obtain $\bm{\rho}^{\bm{\pi}_v}_\xi$ for $\xi\in [0,t]$\;
%$S_v\leftarrow \argmin_{S\in \bigcup_{\xi \in [0,t]}\sweep(\bm{\rho}^{\bm{\pi}_v}_\xi)}\phi_G(S)$\;
%}
%\Return{$S\in \argmin_{v\in V} \phi_G(S_v)$}
\end{algorithm}

Although we cannot solve the differential equation~\eqref{eq:heat-equation-hypergraph} exactly in polynomial time, we can efficiently simulate it by discretizing time using, e.g., the Euler method or the Runge-Kutta method.
Indeed these methods have already been used in practice~\cite{Yoshida:2016ig}.
Alternatively, we can use difference approximation, developed in the theory of monotone operators and evolution equations~\cite{miyadera1992nonlinear}, to obtain the following:
\begin{thm}\label{thm:difference-approximation-intro}
  Let $G=(V,E,w)$ be a weighted hypergraph and $v \in V$, and let $T \geq 1$ and $\lambda \in (0,1)$.
  Then, we can compute (a concise representation) of a solution ${\{\bm{\rho}_t^\lambda\}}_{0 \leq t \leq T}$ such that $\|\bm{\rho}_t^{\bm{\pi}_v} -\bm{\rho}^\lambda_t\|_{D^{-1}} = O(\sqrt{\lambda T})$ for every $0 \leq t \leq T$, in time polynomial in $1/\lambda$, $T$, and $\sum_{e \in E}|e|$.
\end{thm}

%Moreover, we can make the approximate solution arbitrarily close to the true solution by tuning the resolution of discretization (see Lemma~\ref{lem:difference-approximation}).

\subsection{Directed graphs}

We briefly discuss directed graphs here, as we can show analogues of Theorem~\ref{thm:main-intro}, Corollary~\ref{cor:main-intro-cor}, and Theorem~\ref{thm:difference-approximation-intro} for such graphs with almost the same proof.

For a directed graph $G=(V,E,w)$, the \emph{degree} of a vertex $v \in V$ is $d_G(v) = \sum_{e \in E\mid v \in e}w(e)$ and the \emph{volume} of a set $S \subseteq V$ is $\vol_G(S) = \sum_{v \in S}d_G(v)$.
Note that we do not distinguish out-going and in-coming edges when calculating degrees.
Then, the \emph{conductance} of a set $\emptyset \subsetneq S \subsetneq V$ is defined as
\[
  \phi_G(S) := \frac{\min \{ \sum_{e \in \partial^+_G(S)}w(e),\sum_{e \in \partial^-_G(S)}w(e)\}}{\min\{\vol_G(S),\vol_G(V \setminus S) \}},
\]
where $\partial^+_G(S)$ and $\partial^-_G(S)$ are the sets of edges leaving and entering $S$, respectively.
Then, the \emph{conductance} of $G$ is $\phi_G := \min_{\emptyset \subsetneq S \subsetneq V}\phi_G(S)$.
Note that $\phi_G = 0$ when $G$ is a directed acyclic graph.

Yoshida~\cite{Yoshida:2016ig} introduced the notion of a Laplacian for directed graphs and derived Cheeger's inequality, which relates $\phi_G$ and the second-smallest eigenvalue $\lambda_G$ of the normalized Laplacian of $G$.
As with the hypergraph case, computing $\lambda_G$ is problematic, and we can apply an analogue of Theorem~\ref{thm:main-intro} to obtain a set of small conductance without computing $\lambda_G$.
In this paper, we focus on hypergraphs for simplicity of exposition.

\subsection{Sketch of proof}\label{subsec:sketch}

Chung~\cite{Chung:2007ep} presented analogues of
Theorem~\ref{thm:main-intro} and
 Corollary~\ref{cor:main-intro-cor} for usual undirected graphs.
Here, we review her proofs of these analogue, because our proofs of
Theorem~\ref{thm:main-intro} and Corollary~\ref{cor:main-intro-cor} extend them partially.

%the proof of Theorem~\ref{thm:main-intro} presented in Section~\ref{sec:analysis} extends that proof.
%and then we explain how we modify the proof to obtain Theorem~\ref{thm:main-intro}.

For the undirected graph case, we consider the following single-valued differential equation:
\begin{align*}
  \frac{d \bm{\rho}_t}{dt} = - \mathcal{L}_G\bm{\rho}_t \quad \text{and} \quad \bm{\rho}_0 = \bm{s}.
%  \tag{$\mathrm{HE};\bm{s}$}
%  \label{eq:heat-equation-graph}
\end{align*}
This differential equation has a unique global solution $\bm{\rho}_t^{\bm{s}} = \exp(-t\mathcal{L}_G)\bm{s}$.
We define a function $f^{\bm{s}}\colon \mathbb{R}_+ \to \mathbb{R}$ as
\[
  f^{\bm{s}}(t) = \|\bm{\rho}^{\bm{s}}_{t/2} - \bm{\pi} \|_{D^{-1}}^2.
\]
% where $\|\bm{x}\|_{D^{-1}}^2 = \bm{x}^\top D^{-1}\bm{x}$.
When $G$ is connected, $\bm{\rho}^{\bm{s}}_t$ converges to $\bm{\pi}$ as $t\rightarrow\infty$ irrespective of $\bm{s}$; hence, $f^{\bm{s}}$ measures the difference between $\bm{\rho}^{\bm{s}}_{t/2}$ and its unique stationary distribution $\bm{\pi}$.
For a set $S \subseteq V$, we define $\bm{\pi}_S \in \mathbb{R}^V$ as $\bm{\pi}_S(v) = d(v)/\vol(S)$ if $v \in S$ and $\bm{\pi}_S(v) = 0$ otherwise.
%When $S = \{v \}$, we write $\bm{\pi}_v$.
Then, we can show the inequalities
\begin{align}
  \exp(- O(\phi(S)t)) \leq f^{\bm{\pi}_S}(t) \leq \exp\left(-\Omega\left({\left(\kappa^{\bm{\pi}_S}_t\right)}^2t\right)\right),
  \label{eq:key-inequalities-intro}
\end{align}
for every $S\subseteq V$, where $\kappa^{\bm{\pi}_S}_t$ is the minimum conductance of a sweep set with respect to the vector ${(\bm{\rho}^{\bm{\pi}_S}_t(v)/d(v))}_{v \in V}$.
From the closed solution of $\bm{\rho}^{\bm{s}}_t$, we observe that $\bm{\rho}^{\bm{\pi}_S}_{t/2} = \sum_{v \in S}\frac{d (v)}{\vol(S)}\bm{\rho}^{\bm{\pi}_v}_{t/2}$.
Then, we have
\begin{align*}
  & \exp(- O(\phi(S)t)) \leq
  f^{\bm{\pi}_S}(t)
  = \|\bm{\rho}^{\bm{\pi}_S}_{t/2} - \bm{\pi} \|_{D^{-1}}^2
  \leq {\left(\sum_{v \in S} \frac{d(v)}{\mathrm{vol}(S)} \|\bm{\rho}^{\bm{\pi}_v}_{t/2} - \bm{\pi} \|_{D^{-1}}\right)}^2 \tag{by triangle inequality}\\
  & \leq
  \max_{v \in S} \|\bm{\rho}^{\bm{\pi}_v}_{t/2} - \bm{\pi} \|_{D^{-1}}^2
  =
  \max_{v \in S} f^{\bm{\pi}_v}(t)
  \leq \max_{v \in S} \exp\left(-\Omega\left({\left(\kappa^{\bm{\pi}_v}_t\right)}^2 t\right)\right).
\end{align*}
Taking the logarithm yields the desired result.

The main obstacle to extending the above argument to hypergraphs is that $\bm{\rho}_t$ does not have a closed-form solution as $\mathcal{L}_G$ is no longer a linear operator and single-valued.
To overcome this obstacle, we observe that there exists the sequence $t_0 = 0 < t_1 < t_2 < \cdots$
% with $\lim_{i \to \infty }t_i= \infty$ exists,
such that $\mathcal{L}_G$ can be regarded as a linear operator
$\mathcal{L}_i$
in each interval $[t_i,t_{i+1})$.
Here, $\mathcal{L}_i$ is the normalized Laplacian of a graph constructed from the hypergraph $G$ and the vector
$\bm{\rho}_{t_i}$.
Then, we can show a counterpart of the second inequality of~\eqref{eq:key-inequalities-intro} for each $f_i^{\bm{s}}\colon \mathbb{R}_+ \to \mathbb{R}$ defined as $f_i^{\bm{s}}(\Delta) = \|\bm{\rho}^{\bm{s}}_{t_i+\Delta/2}-\bm{\pi} \|^2_{D^{-1}}$, which is sufficient for our analysis.
(We will use another equivalent definition for $f_i^{\bm{s}}$ for convenience. See Section~\ref{sec:analysis} for details.)

Another obstacle is that the triangle inequality applied in the above argument is not true in general, because $\bm{\rho}^{\bm{\pi}_S}_{t/2}$ may not generally be equal to $\sum_{v \in S}\frac{d(v)}{\vol(S)}\bm{\rho}^{\bm{\pi}_v}_{t/2}$ for the hypergraph case.
Due to this obstacle, it is hard to obtain a counterpart of the first inequality of~\eqref{eq:key-inequalities-intro}.
To overcome this problem, using the fact that the logarithmic derivative $g_v(t)$ is monotonically non-increasing and considering $t>T$ for $T >0$,
we obtain a non-trivial lower bound $\exp(-O (g_v(T)(t-T))$ of the
square of norm
$\|\bm{\rho}^{\bm{\pi}_v}_{t} - \bm{\pi} \|_{D^{-1}}^2$.
Then, we show that $g_v(T)$ goes to an eigenvalue of the normalized
Laplacian $\ca{L}_{G_{v,T}}$ as $T$ becomes larger.
Hence, if $g_v(T)$ is close to $\lambda_2(G_{v,T})$, by using the Cheeger inequality~\eqref{eq:cheeger-graph} for graphs and the
relation $\phi_{G_{v,T}} \leq \phi_G$, we obtain a counterpart of the first inequality of~\eqref{eq:key-inequalities-intro}.
%To derive the triangle inequality, we exploit the theory of maximal monotone operators and evolution equations~\cite{miyadera1992nonlinear} and borrow the concept of difference approximation of the solution.

\subsection{Related work}
As noted above, an analogue of Theorem~\ref{thm:main-intro} for usual graphs has been presented by Chung~\cite{Chung:2007ep}.
However, as the normalized Laplacian $\mathcal{L}_G = I-A_G D_G^{-1}$ is a matrix for the graph case, that analysis is much simpler than that presented herein.
Kloster and Gleich~\cite{Kloster:2014wq} have presented a deterministic algorithm that approximately simulates the heat equation for graphs. Hence, they extracted a tightly connected subset by considering a local part of the graph only.

The concept of the Laplacian for hypergraphs has been implicitly employed in semi-supervised learning on hypergraphs in the form $\bm{x}^\top L_G(\bm{x})$, where $\mathcal{L}_G(\bm{x}) = L_G(D_G^{-1}\bm{x})$~\cite{Hein:2013wc,Zhang:2017va}. This concept was then formally presented by Chan~et~al.~\cite{Chan:2018eu} at a later time.
Subsequently, the Laplacian concept was further generalized to handle submodular transformations~\cite{Li:2018we,Yoshida:2017zz}; this development encompasses Laplacians for graphs, hypergraphs~\cite{Chan:2018eu}, directed graphs~\cite{Yoshida:2016ig}, and directed hypergraphs~\cite{chan2017diffusion}. On our work here, we need
precise description of undirected graphs $\wt{G}_i$ introduced below. To achieve this,
we borrow some results in~\cite[Sections~3 and~4]{chan2017diffusion}.

Finally, we note that another type of Laplacian for hypergraphs, which essentially replaces each hyperedge with a clique, has been used in the literature~\cite{Agarwal:2006ci,Scholkopf:2006vj}.
We stress that that Laplacian differs from the Laplacian for hypergraphs studied in this work.

% of Chung considered the following differential equation:
% \begin{align}
%   \frac{d \bm{\rho}}{dt} = -L\bm{\rho}, \label{eq:heat-equation-graph}
% \end{align}
% where $W_G = L_G D_G^{-1} = I - A_G D_G^{-1}$ is the transition matrix of a random walk.

% She showed a local versions of the Cheeger inequality, which relates the conductance of a subset to the heat kernel with seeds as the vertices in the subset.
% \begin{thm}\label{thm:chung-intro}
%   In a graph $G = (V,E)$, for a subset $S \subseteq V$ of vertices in $G$ with $\vol(S) \leq \vol(G)/2$ and a real value $t \geq 0$, the conductance of $S$ satisfies the following:
%   \[
%     \frac{\phi_G(S)}{1-\vol(S)/\vol(G)} \geq \frac{\kappa^2_{t,S}}{2}-\frac{1+\log|S|}{t}.
%   \]
%   where $\kappa_{t,S}$ denote the minimum conductance over all sweep sets of $\bm{\rho}_{t/2,u}$ for all $u \in S$.
% \end{thm}

\subsection{Organization}
The remainder of this paper is organized as follows.
In Section~\ref{sec:pre}, we introduce the basic concepts used throughout this paper.
In Section~\ref{sec:heat-equation}, we show some basic facts on the heat equation~\eqref{sec:heat-equation}.
In Section~\ref{sec:analysis}, we prove Theorem~\ref{thm:main-intro}.
 We show that~\eqref{eq:heat-equation-hypergraph} has a unique global solution in Section~\ref{sec:solution}.
% In Section~\ref{sec:triangle-inequality}, we prove the triangle inequality discussed in Section~\ref{subsec:sketch}.
A proof of Theorem~\ref{thm:difference-approximation-intro} is given in Section~\ref{sec:difference-approximation}.

%!TEX root=./main.tex

\section{Preliminaries}\label{sec:pre}

For a vector $\bm{x} \in \mathbb{R}^V$ and a set $S \subseteq V$, let $\bm{x}(S) = \sum_{v \in S}\bm{x}(v)$.
For a vector $\bm{x} \in \mathbb{R}^V$ and a positive semidefinite matrix $A \in \mathbb{R}^{V \times V}$, we define $\langle \bm{x},\bm{y}\rangle_A = \bm{x}^\top A\bm{y}$ and $\|\bm{x}\|_A= \sqrt{\langle \bm{x},\bm{x}\rangle_A} =  \sqrt{\bm{x}^\top A \bm{x}}$.

Let $G=(V,E,w)$ be a hypergraph.
We omit the subscript $G$ from notations such as $A_G$ when it is clear from the context.
For a set $S \subseteq V$, let $\bm{1}_S \in \mathbb{R}^V$ denote the characteristic vector of $S$, that is, $\bm{1}_S(v) = 1$ if $v \in S$ and $\bm{1}_S(v) = 0$ otherwise.
When $S=V$ or $S=\{v \}$, we simply write $\bm{1}$ and $\bm{1}_v$, respectively.
For a set $S \subseteq V$, we define a vector $\bm{\pi}_S \in \mathbb{R}^V$ as $\bm{\pi}_S(v) = \frac{d_G(v)}{\vol_G(S)}$ if $v \in S$ and $\bm{\pi}_S(v) = 0$ otherwise.
When $S=V$ or $S=\{v \}$, we simply write $\bm{\pi}$ and $\bm{\pi}_v$, respectively.
For a vector $\bm{\rho} \in \mathbb{R}^V$, we write $\bm{\rho}/d_G$  to denote a vector with $(\bm{\rho}/d_G)(v) = \bm{\rho}(v)/d_G(v)$ for each $v \in V$.

\subsection{Normalized Laplacian for hypergraphs}\label{subsec:laplacian-for-hypergraphs}

We define (random-walk) normalized Laplacian for hypergraphs precisely.
Let $G=(V,E,w)$ be a hypergraph.
For each hyperedge $e \in E$, we define a polytope $B_e = \mathrm{conv}(\{ \bm{1}_u - \bm{1}_v \mid u, v \in e \})$, where $\mathrm{conv}(S)$ denotes the convex hull of $S \subseteq \mathbb{R}^V$.
Then, the \emph{Laplacian} $L_G\colon \mathbb{R}^V \to 2^{\mathbb{R}^V}$ of $G$ is defined as
\begin{align}
  L_G(\bm{x}) = \left\{ \sum_{e \in E}w(e) \bm{b}_e\bm{b}_e^\top \bm{x} \mid \bm{b}_e \in \argmax_{\bm{b} \in B_e}\bm{b}^\top \bm{x}  \right\},
  \label{eq:laplacian}
\end{align}
and the \emph{normalized Laplacian} is defined as $\mathcal{L}_G\colon \bm{x} \mapsto L_G(D_G^{-1}\bm{x})$.

We can write $L_G(\bm{x})$ more explicitly as follows.
For each hyperedge $e \in E$, let $S_e = \argmax_{v\in e}\bm{x}(v)$ and $I_e = \argmin_{v\in e}\bm{x}(v)$.
Let $E' = \{uv \mid e \in E, u \in S_e, v \in I_e \} \cup \{vv \mid v \in V \}$.
Then, we arbitrarily define a function $w'_e\colon E' \to \mathbb{R}_+$ so that $w'_e(uv) > 0$ only if $u \in S_e$ and $v \in I_e$ and we have $\sum_{u \in S_e, v \in I_e}w'_e(uv) = w(e)$.
Then, we construct a graph $G' = (V,E',w')$, where $w'(uv) = \sum_{e \in E\mid u \in S_e, v \in I_e}w'_e(uv)$ for each $uv \in E'$ and $w'(vv) = d_G(v) - \sum_{e \in E' \mid v \in e}w'(e)$ for each $v \in V$.
Note that $d_G(v) = d_{G'}(v)$ for every $v \in V$.
Let $\mathcal{G}(G,\bm{x})$ be the set of graphs constructed this way.
Then, we have $L_G(\bm{x}) = \{ L_{G'}\bm{x} \mid  G' \in \mathcal{G}(G,\bm{x}) \}$.

% Note that $L_G(\bm{x})$ is solely determined by the ordering of elements in $\bm{x}$.
% When elements in $\bm{x}$ is pairwise distinct, we construct a weighted graph $G_{\bm{x}}=(V,E_{\bm{x}},w_{\bm{x}})$ from $G$ by replacing each hyperedge $e \in E$ with an edge $uv$, where $u = \argmin_{w \in e}\bm{x}(w)$ and $v = \argmin_{w \in e}\bm{x}(w)$. \ynote{fix this part}
% Then, we have $L(\bm{x}) = L_{\bm{x}}\bm{x}$.

We can understand Laplacian for hypergraphs in terms of submodular functions.
Let $F_e\colon 2^V \to \{0,1 \}$ be the cut function associated with a hyperedge $e \in E$, that is, $F_e(S) = 1$ if and only if $S \cap e \neq \emptyset$ and $(V \setminus S ) \cap e \neq \emptyset$.
It is known that $F_e$ is \emph{submodular}, that is, $F_e(S)+F_e(T) \geq F_e(S \cap T)+F_e(S \cup T)$ holds for every $S, T \subseteq V$.
Then, $B_e$ is the base polytope of $F_e$ and $\bm{b}_e$ in~\eqref{eq:laplacian} is chosen so that $\bm{b}_e^\top \bm{x} = f_e(\bm{x})$, where $f_e\colon \mathbb{R}^V \to \mathbb{R}$ is the \emph{Lov{\'a}sz extension} of $F_e$.
See~\cite{Fujishige:2005uc} for detailed definitions of these notions.

When $G=(V,E,w)$ is a usual graph, its \emph{Laplacian} $L_G \in \mathbb{R}^{V \times V}$  and \emph{(random-walk) normalized Laplacian} $\mathcal{L}_G \in \mathbb{R}^{V \times V}$ are defined as $D_G-A_G$ and $I_G-A_G D_G^{-1}$, respectively.
Indeed, this coincides with~\eqref{eq:laplacian} when we regard $G$ as a hypergraph with each hyperedge having size two.
% With this observation, Yoshida~\cite{Yoshida:2017zz} showed that the notion of Laplacian~\eqref{eq:laplacian} can be generalized to any \emph{submodular transformations} $F\colon 2^V \to \mathbb{R}^E$ with $F(\emptyset)=F(V)=\bm{0}$, that is, each $F_e\colon S \mapsto F(S)(e)$ is a submodular function.
% Although most of our argument is valid for general submodular transformations, in this work, we focus on hypergraphs for simplicity.

%!TEX root=./main.tex

\section{Properties of Solutions to Heat Equation}\label{sec:heat-equation}
We review some facts on the heat equation~\eqref{eq:heat-equation-hypergraph}.
We say that ${\{\bm{\rho}_t\}}_{t \geq 0}$ is a \emph{solution} of~\eqref{eq:heat-equation-hypergraph} if $\bm{\rho}_t$ is absolutely continuous with respect to $t$ (hence $\bm{\rho}_t$ is differentiable at almost all $t$) and $\bm{\rho}_0 = \bm{s}$ and satisfies $\frac{d}{dt} \bm{\rho}_t \in -\mathcal{L}_G (\bm{\rho}_t)$ for almost all $t \geq 0$.
As we see in Section~\ref{sec:solution}, the heat equation~\eqref{eq:heat-equation-hypergraph} always has a unique global solution.
Also as we mentioned, when $G$ is connected, $\bm{\rho}_t$ converges to $\bm{\pi}$ as $t \to \infty$ for any $\bm{s}\in \mathbb{R}^V$ with
$\sum_{v \in V}\bm{s}(v)=1$.

We consider the heat equation~\eqref{eq:heat-equation-hypergraph} on a hypergraph $G= (V, E,w)$ with an initial vector $\bm{s} \in \mathbb{R}^V$ and
let ${\{\bm{\rho}_t^{\bm{s}}\}}_{t \geq 0}$ be its unique solution.
Let $\bm{\mu}^{\bm{s}}_t = D^{-1}\bm{\rho}^{\bm{s}}_t $.
Then, there is an ordered equivalence relation $(\sigma^\ast, \succ)$ on $V$ consistent with ${\{ d^k\bm{\mu}^{\bm{s}}_t/dt^k \}}_k$ introduced in~\cite[Section~3.1]{chan2017diffusion}, i.e., for $u,v \in V$,
$u \sim_{\sigma^\ast} v $ if all higher (right) derivatives of
$\bm{\mu}^{\bm{s}}_t (u)$ and $\bm{\mu}^{\bm{s}}_t (v)$ at $t = 0$ are equal and
for two $\sigma^\ast$-equivalence classes $U$ and $U'$, $U \succ U'$
if there is an integer $l \in \ZZ_+$ such that for $u \in U$ and $u' \in U'$,
the following hold:
\begin{align*}
\left.\frac{d^{k}\bm{\mu}_{t}^{\bm{s}}}{dt^k}\right|_{t=0}(u) = \left.\frac{d^{k}\bm{\mu}_{t}^{\bm{s}}}{dt^k}\right|_{t=0}(u')  \text{ for } k=0, \dots, l-1, \text{ and }
 \left.\frac{d^{l}\bm{\mu}_{t}^{\bm{s}}}{dt^l}\right|_{t=0}(u)  > \left.\frac{d^{l}\bm{\mu}_{t}^{\bm{s}}}{dt^l}\right|_{t=0}(u').
\end{align*}
We define $\succeq$ as $\succ$ or $=$.
We divide $V$ by the equivalence relation $\sigma^\ast$ as $V = \bigsqcup_{k=1}^m U_k$ so that $U_k \succ U_{k+1}$ for every $k$.
For $v \in V$, let ${[v]}_{\sigma^\ast}$ be the equivalence class including $v$.

Let $\bm{x} = \bm{\mu}_{0}^{\bm{s}} = D^{-1}\bm{s}$.
For $e\in E$, we recall
$S_e = S_e(\bm{x}) = \argmax_{v\in e} \bm{x}(v)$ and $I_e = I_e(\bm{x}) = \argmin_{v\in e} \bm{x}(v)$.
Then, we define $S_e^{\sigma^\ast} = S_e^{\sigma^\ast}(\bm{x})$ and
$I_e^{\sigma^\ast} = I_e^{\sigma^\ast}(\bm{x})$ as
\begin{align*}
  S_e^{\sigma^\ast}(\bm{x}) &= \{ u \in S_e(\bm{x}) \mid {[u]}_{\sigma^\ast} \succeq  {[v]}_{\sigma^\ast} \text{ for any } v \in  S_e(\bm{x})  \} \\
  I_e^{\sigma^\ast}(\bm{x}) &= \{ u \in I_e(\bm{x}) \mid {[u]}_{\sigma^\ast} \preceq  {[v]}_{\sigma^\ast} \text{ for any }  v \in  I_e(\bm{x})  \}.
\end{align*}
We set $\wt{V}$ as a complete system of representatives $\{ u_1, u_2, \ldots, u_m \} \subseteq V$ ($u_k \in U_k$) and set $\wt{E} = \{ u_k u_l \subset \wt{V} \mid k, l = 1\dots m \}$.
We define the weights $\wt{w}$ on $\wt{E}$ as
\begin{align*}
 \wt{w}(u_k u_l) &= \sum_{e\in E, \ S_e^{\sigma^\ast}\cap U_k \neq \emptyset \atop
 I_e^{\sigma^\ast}\cap U_l \neq \emptyset} w_e +
 \sum_{e\in E, \ S_e^{\sigma^\ast}\cap U_l \neq \emptyset \atop
 I_e^{\sigma^\ast}\cap U_k \neq \emptyset} w_e \ \ \ \mathrm{for}  \ k \neq l, \\
  \wt{w}(u_k u_k) &=  \sum_{v \in U_k} d_G(v) -
\sum_{l=1, \dots, m, \atop l\neq k} \wt{w}(u_k u_l),
 %\sum_{e \in E \atop (S_e^{\sigma^\ast} \cup I_e^{\sigma^\ast}) \cap U_k \neq \emptyset} w_e,
\end{align*}
%where $\wt{w}_k = \sum_{v \in U_k} w_v$.
Then, the triple $\wt{G} = (\wt{V}, \wt{E}, \wt{w})$ can be regarded as a weighted undirected graph.
Let $d_{\wt{G}}(u_k) = \sum_{v \in U_k} d_G(v)$,  $D_{\wt{G}} = \mathrm{diag}(d_{\wt{G}}(u_k)) \in \RR^{m\times m}$,
$\wt{ \bm{s}} = {\left(\sum_{u \in U_k}\bm{s}(u)\right)}_k \in \RR^m$, and $\wt{\bm{x}} = {(\bm{x}(u_k))}_k \in \RR^m$. Then, we have $\wt{\bm{s}} = \wt{D} \wt{\bm{x}}$ by using
the equivalence relation $\sigma^\ast$. Then for
$\wt{\bm{\rho}}_t^{\bm{s}} = {\left( \sum_{u \in U_k}\bm{\rho}_t^{\bm{s}}(u)\right)}_k \in \RR^m$, the following holds.
The proof is deferred to Appendix~\ref{sec:proof-of-graph-diffusion}.
\begin{thm}\label{thm:graph-diffusion}
$\wt{\bm{\rho}}_t^{\bm{s}}$ is a unique solution of the following heat equation:
\[
 \frac{d \wt{\bm{\rho}}_t}{dt} = - \mathcal{L}_{\wt{G}}
 \wt{\bm{\rho}}_t, \ \ \wt{\bm{\rho}}_0 = \wt{\bm{s}}.
\]
until when a next tie occurs for $\bm{\mu}_t^{\bm{s}}$, i.e., if we retake
the ordered equivalence relation $\sigma^\ast$ consistent with
${\{ d^k \bm{\mu}^{\bm{s}}_t /dt^k \}}_k$ at $t$, either $S_e^{\sigma^\ast}(\bm{\mu}_t^{\bm{s}})$ or $I_e^{\sigma^\ast}(\bm{\mu}_t^{\bm{s}})$ changes for some $e$. Here, $\ca{L}_{\wt{G}}$ is the graph normalized Laplacian of $\wt{G}$.
Moreover, the solution of this heat equation $\wt{\bm{\rho}}_t^{\bm{s}}$ determines
$\bm{\rho}_t^{\bm{s}}$ for such $t$.
\end{thm}

%This means that $\wt{\rho_{t, \bm{s}}}$ diffuses as a usual diffusion on an undirected graph
% until the ordered equivalence relation $\sigma^\ast$ changes.

Then, there is a time sequence $t_{0}= 0 < t_{1}< t_{2} < \cdots $ such that there is a weighted graph $\wt{G}_i = (\wt{V}_i, \wt{E}_i, \wt{w}_i) $ for each $i \in \mathbb{Z}_+$ such that the heat equation on the interval $[t_{i}, t_{i+1})$ satisfies
\[
  \frac{d \wt{\bm{\rho}}_t}{d t} = - \mathcal{L}_i \wt{\bm{\rho}}_t,
\]
where $\wt{\bm{\rho}}_{t} :=
{\left(\sum_{u \in U^i_k}\bm{\rho}_{t}(u)\right)}_k$ for equivalence classes ${\{ U^i_k \}}_k$ by
equivalence relation $\sigma^\ast$ consistent with
${\{ d^k \bm{\mu}^{\bm{s}}_t/dt^k \}}_k$ at $t=t_i$, and
$\mathcal{L}_i$ is the normalized Laplacian associated with $\wt{G}_i$.
Hence, we can write the solution $\wt{\bm{\rho}}_{i,\Delta} :=
{\left(\sum_{u \in U^i_k}\bm{\rho}_{t_i + \Delta}(u)\right)}_k$ for $\Delta \in [0,t_{i+1}-t_i)$ as
\begin{align}
  \wt{\bm{\rho}}_{i,\Delta}:= H_{i,\Delta} \wt{\bm{\rho}}_{t_{i}},
  \quad \text{where }
  H_{i,\Delta} := e^{-\Delta \mathcal{L}_i} = \sum_{n = 0}^\infty \frac{{(-\Delta)}^n \mathcal{L}_i^n}{n!}.
  \label{eq:rho-i-Delta}
\end{align}
%For $t \in [t_{i},t_{i+1})$, it is easy to see that
%\[
%  \bm{\rho}_t = \bm{\rho}_{i,t - t_i} = H_{i,t-t_i}\bm{\rho}_{t_i}  =  H_{i,t-t_i} H_{i-1,t_{i} - t_{i-1}} \cdots H_{1,t_{2} - t_{1}}H_{0,t_{1}} \bm{s}.
%\]
Although $\wt{\bm{\rho}}_{i,\Delta}$ was originally defined for $\Delta \in [0,t_{i+1}-t_i)$, we can extend it to any $\Delta \geq 0$ by using~\eqref{eq:rho-i-Delta}.
When we want to stress the initial vector, we write
$\bm{\rho}_t^{\bm{s}}$, $\bm{\rho}_{i,\Delta}^{\bm{s}}$, $\wt{\bm{\rho}}_t^{\bm{s}}$, $\wt{\bm{\rho}}_{i,\Delta}^{\bm{s}}$, etc.

In what follows, we assume that for the initial vector $\bm{s}$ and $t > 0$,
there is an integer $n\in\ZZ_+$ and a sequence $0=t_0 < t_1 < \cdots < t_n < T$ satisfying the following condition:
On each interval $[t_i, t_{i+1}]$, $i = 0, 1, \dots, n-1$, and $[t_n, T]$,
the solution $\bm{\rho}^{\bm{s}}_t$ of heat equation~\eqref{eq:heat-equation-hypergraph} can be obtained
by the solution of the heat equation on the weighted graph $\wt{G}_i$ ($i=0, 1, \dots, n$) as above.
We assume this only for simplicity of exposition.
Indeed, if the above condition does not hold, the sequence ${\{ t_i \}}_i$ converges to some $T_0 < \infty$.
Then, the existence of the global solution $\bm{\rho}_{t}$ shown in Section~\ref{sec:solution} implies that
$\bm{\rho}_{T_0}$ is well defined, and hence another sequence ${\{t'_i \}}_i$ starts from $T_0$ again, we can continue this process until we reach $T$.
% Moreover, the solution $\bm{\rho}_{t}$ should be constructed by repeating this step (possibly countably infinitely many times), from the global solu
% If the sequence stops at some $T' < \infty$, this is a contradiction of the existence of $\rho_t$.
It is not hard to generalize our argument for such a case.
%By using this fact, we can easily genearlize our argument.
% for the general case.
% and slightly modifying the proof, we can generalize
% a similar theorems in this paper for this case.

% We remark two points to use this definition:
% \begin{enumerate}
%  \item We can regard the heat equation as on undirected usual graph only on each interval
%  $[t_i, t_{i+1})$.
%  \item We have to take care when $t$ goes through different intervals.
%  In particular, because the undirected graph $G_i$ is determined by $f$, if we take another initial value $g_i$ at $t = t_i$, the heat $\rho^i_{\Delta,g_i} = H^i_\Delta g_i$ does not have any information before $t=t_i$. However, the heat $\rho^i_{\Delta,g_i} = H^i_\Delta g_i$ means the heat
%  on $G_i$ at time $\Delta$ starting from $g_i$ at $\Delta = 0$.
% \end{enumerate}

% To generalize the result of Fan Chung \cite{chung2007heat}, the important thing is to analyze
% this summation:
% $$
%  - \sum_{u \sim v} \left( \frac{\rho^i_{\Delta/2, f_i}(u)}{d_u} -\frac{\rho^i_{\Delta/2, f_i}(v)}{d_v} \right)^2w^i_{u,v},
% $$
% because this summation relates to the conductance given by sweep. We should consider the problem with taking case of this form.

%!TEX root=./main.tex

\section{Proof of Theorem~\ref{thm:main-intro}}\label{sec:analysis}

% For a vertex $u \in V$, let $t_1,\ldots$ be sequence of switching time.
In this section, we prove Theorem~\ref{thm:main-intro}.
Missing proofs are found in Appendix~\ref{sec:proofs-of-useful-lemmas}.

% In what follows, we fix a hypergraph $G=(V,E)$ and a set $S \subseteq V$ with $\vol(S) \leq \vol(V)/2$.
% We borrow notations such as $t_i$, $G_i$, and $\bm{s}_i$ from Section~\ref{subsec:heat-equation} and we simply write $\wt{\bm{\rho}}_{i,\Delta}$ instead of $\bm{\rho}^i_{\Delta,\bm{s}_i}$.
% We define $\tilde{\bm{\rho}}^i_\Delta = D_{\tilde{G}_i}^{-1}\wt{\bm{\rho}}_{i,\Delta}$ as the normalized version of $\wt{\bm{\rho}}_{i,\Delta}$.

Consider the heat equation~\eqref{eq:heat-equation-hypergraph}.
We borrow notations from Section~\ref{sec:heat-equation}.
For each $i \in \mathbb{Z}_+$, we define a function $f_i\colon \mathbb{R}_+ \to \mathbb{R}$ as
\[
  f_i(\Delta) := \wt{\bm{\rho}}_{i,0}^\top D_{\tilde{G}_i}^{-1} \left(\wt{\bm{\rho}}_{i,\Delta} - \wt{\bm{\pi}}^i\right),
\]
where $\wt{\bm{\pi}}^i = {\left(\sum_{u \in U^i_k} \bm{\pi}(u)\right)}_k = {\left(d_{\wt{G}_i}(u_k^i)/\vol \wt{V}_i\right)}_k$.
When we wish to stress the initial vector $\bm{s} \in \mathbb{R}^V$, we write $f^{\bm{s}}_i$.
As the following proposition implies, the value of $f_i(\Delta)$ indicates the difference between $\wt{\bm{\rho}}_{i,\Delta/2}$ and the stationary distribution $\wt{\bm{\pi}}^i$
on $\wt{G}_i$.
\begin{prop}\label{prop:norm}
  For any initial vector $\bm{s} \in \mathbb{R}^V$, $i \in \mathbb{Z}_+$, and $\Delta \geq 0$, we have
  \[
    f_i(\Delta) = \|\wt{\bm{\rho}}_{i,\Delta/2} - \wt{\bm{\pi}}^i\|_{D_{\tilde{G}_i}^{-1}}^2 = \sum_{v \in \wt{V}_i} {\left( \frac{\wt{\bm{\rho}}_{i,\Delta/2}(v)}{d_{\wt{G}_i}(v)} -\frac{1}{\vol(\wt{V}_i)}\right)}^2 d_{\wt{G}_i}(v) \geq 0.
  \]
\end{prop}
\begin{proof}
  We have
  %The proof of this proposition is immediate:
  \[
    \wt{\bm{\rho}}_{i,0}^\top D_{\tilde{G}_i}^{-1}(\wt{\bm{\rho}}_{i,\Delta}-\wt{\bm{\pi}}^i) =
    \| D_{\tilde{G}_i}^{-1/2} (H_{i,\Delta/2} - \wt{\bm{\pi}}^i \1^\top)\wt{\bm{\rho}}_{i,0}\|^2
    = \| D_{\tilde{G}_i}^{-1/2} (\wt{\bm{\rho}}_{i,\Delta/2} - \wt{\bm{\pi}}^i)\|^2.
    \qedhere
  \]
  %The second equality is obtained through a direct calculation.
\end{proof}

The following lemma shows the compatibility of norms between vectors on $G$ and $\wt{G}_i$.
\begin{lem}\label{lem:equality-of-norms}
 For $t = t_i + \Delta$, $0 \leq \Delta\leq t_{i+1}- t_i$, we have
  $\| \wt{\bm{\rho}}_{i, \Delta} -\wt{\bm{\pi}}^i \|_{D_{\tilde{G}_i}^{-1}} =
    \| \bm{\rho}_t -\bm{\pi} \|_{D^{-1}}$.
\end{lem}

Theorem~\ref{thm:main-intro} is obtained by bounding $f_i(\Delta)$ from above and below.
To obtain an upper bound, for $0\leq T \leq t$, we define
\begin{align*}
  &\wt{\kappa}_{i,I}  =  \min\left\{ \phi_{\wt{G}_i}(S) \ \left| \ 
   \xi \in I, \, S \in \sweep\left(\frac{\wt{\bm{\rho}}_{i,\xi}}{d_{\wt{G}_i}}\right) \right. \right\} \quad ( i \in \mathbb{Z}_+, I \subset [0, t_{i + 1}-t_i]), \\
  &\wt{\kappa}_i  = \wt{\kappa}_{i,[0,t_{i+1}-t_i]} \quad (i \in \mathbb{Z}_+), \\
%  &\wt{\kappa}_t  = \min\left\{ \min_{j=0,\dots,i-1} \wt{\kappa}_j,\, \wt{\kappa}_{i,\Delta} \right\},\\
&\wt{\kappa}_{T,t}  = \min\left\{ \min_{j=i_0+1,\dots,i_1-1} \wt{\kappa}_j,\, \wt{\kappa}_{i_0, [T-t_{i_0}, t_{i_0+1}]},  \wt{\kappa}_{i_1,[0, t- t_{i_1}]} \right\},\\
  &\text{where } i_0,i_1 \in \mathbb{Z}_+ \text{ are such that }
  T \in [t_{i_0},t_{i_0+1}) \text{ and } 
  t \in [t_{i_1},t_{i_1+1}). 
  %\wt{\kappa}_t & = \min\left\{ \min_{0 \leq j \leq i-1} \wt{\kappa}_j,\wt{\kappa}_{i,\Delta} \right\}, \text{where } i \in \mathbb{Z}_+ \text{ is such that } t \in [t_i,t_{i+1})\text{ and } \Delta = t-t_i. \quad (t \geq 0)
\end{align*}
Again, when we wish to stress the initial vector $\bm{\pi}_v \in \mathbb{R}^V$, we write $\wt{\kappa}_{i,I}^{v}$, etc. In the following lemma, we present an upper bound on a quotient of norms of heat when the initial vector $\bm{s}$ is $\bm{\pi}_v$ for some $v\in V$.
\begin{lem}\label{lem:upper-bound}
%  Consider the heat equation $(\mathrm{HE}; \bm{\pi}_S)$ for a set $S \subseteq V$.
%  For $t\geq 0$, let $i \in \mathbb{Z}_+$ be such that $t\in [t_i, t_{i+1})$ and let $\Delta = t - t_i$ and $t' = t_i + \Delta/2$.
%  Then, we have
For any $t \geq T \geq 0$, the following inequality holds:
  \[
   %   f_i(\Delta) \leq \frac{1}{\vol(S)} (1-\wt{\bm{\pi}}^i(S))\exp\left(-\Delta \frac{{\left(\wt{\kappa}_{i,\Delta}\right)}^2}{2}\right),
   \frac{\| \bm{\rho}_t^{\bm{\pi}_v} - \bm{\pi} \|_{D^{-1}}^2}{
 \| \bm{\rho}_T^{\bm{\pi}_v} - \bm{\pi} \|_{D^{-1}}^2}
\leq \exp(- {(\wt{\kappa}_{T,t}^v)}^2(t-T)).
%    f_i(\Delta) \leq \frac{1-\bm{\pi}(S)}{\vol(S)}\exp\left(- \wt{\kappa}_{t'}^2t'\right).
  \]
%   Here, $\wt{\kappa}_{t, \bm{s}}$ is the minimum among the minimum conductance of a sweep set of $G_i$ with respect to the function $v \mapsto \wt{\bm{\rho}}_{i,\Delta}(v)/d(v)$ and
%  the minimum conductances of a sweep set of $G_j$ with respect to the function $v \mapsto \bm{\rho}^j_{t_j-t_{j-1}}(v)/d(v)$ for $j=0, \dots, i-1$.
\end{lem}

Next, we consider a lower bound on the squared norm of the heat with the initial vector $\pi_v$.
Let $T \geq 0$ and set
\[
 g_v(T) = -\left.\frac{d}{dt} \log \| \bm{\rho}_t^{\bm{\pi}_v} - \bm{\pi} \|_{D^{-1}}^2 \right|_{t=T} =
 -\frac{\frac{d}{dt}\| \bm{\rho}_t^{\bm{\pi}_v} - \bm{\pi} \|_{D^{-1}}^2|_{t=T}}
 {\| \bm{\rho}_T^{\bm{\pi}_v} - \bm{\pi} \|^2_{D^{-1}}} 
 = 2\frac{\langle 
 \bm{\rho}_T^{\bm{\pi}_v}, \ca{L}(\bm{\rho}_T^{\bm{\pi}_v})\rangle_{D^{-1}}}
 {\| \bm{\rho}_T^{\bm{\pi}_v} - \bm{\pi} \|^2_{D^{-1}}}.
\]
Then, the following inequality holds:
\begin{lem}\label{lem:lower-bound}
For any $t\geq T$, the following inequality holds:
\[
 \frac{\| \bm{\rho}_t^{\bm{\pi}_v} - \bm{\pi} \|_{D^{-1}}^2}{
 \| \bm{\rho}_T^{\bm{\pi}_v} - \bm{\pi} \|_{D^{-1}}^2}
 \geq \exp\left( - g_v(T)(t-T) \right).
\]
%  Consider the heat equation $(\mathrm{HE}; \bm{\pi}_S)$ for a set $S \subseteq V$.
%  For $t\geq 0$, let $i \in \mathbb{Z}_+$ be such that $t\in [t_i, t_{i+1})$ and let $\Delta = t - t_i$ and $t' = t_i + \Delta/2$.
%  Then, we have
%  \[
%    f_i(\Delta) \geq \frac{1-\bm{\pi}(S)}{\vol(S)} \exp\left(- \frac{2\phi_G(S)}{1-\bm{\pi}(S)}t'\right).
%  \]
\end{lem}

%Next, we consider a lower bound on $f_i(\Delta)$, when the initial vector $\bm{s}$ is $\bm{\pi}_S$ for some set $S \subseteq V$.

%The following lemma is useful to relate the heat equation solutions to the
%different initial vectors. %First, we present the following lemma:
%\begin{lem}\label{lem:triangle-inequality}
%  Let $\bm{s}_1,\ldots,\bm{s}_m \in \mathbb{R}^V$ be vectors and let $\bm{s} = \sum_{i=1}^m \bm{s}_i$.
%  Then, we have
%  \[
%    \left\| \bm{\rho}_t^{\bm{s}} \right\|_{D^{-1}} \leq  \sum_{i=1}^m \| \bm{\rho}_t^{\bm{s}_i} \|_{D^{-1}}.
%  \]
%\end{lem}
%The proof requires difference approximation, introduced in Section~\ref{sec:triangle-inequality}, and the proof is deferred there.
%% We provide proofs of Lemmas~\ref{lem:upper-bound},~\ref{lem:lower-bound},~\ref{lem:equality-of-norms}, and~\ref{lem:triangle-inequality} in Sections~\ref{subsec:upper-bound},~\ref{subsec:lower-bound},~\ref{subsec:equality-of-norms}, and~\ref{sec:triangle-inequality}, respectively.

Based on these lemmas, we obtain the following:
\begin{thm}\label{thm:main-S}
  Let $G=(V,E,w)$ be a hypergraph, and $v \in V$ and $t \geq T \geq 0$.
  Then, we have
  \[
    g_v(T) \geq {(\wt{\kappa}_{T,t}^v)}^2.
  \]
\end{thm}

We remark that the $g_v(T)$ is close to an eigenvalue of normalized 
Laplacian of an undirected graph in the following sense: 
For $\bm{\rho}^{\bm{\pi}_v}_{t}$, as in Section \ref{subsec:laplacian-for-hypergraphs}, 
there is a graph $G_{v,t} = G_{\bm{\rho}^{\bm{\pi}_v}_{t}} = 
(V, E_{v,t}, w_{v,t})$ such that 
$\ca{L}_{G_{v,t}} \bm{\rho}^{\bm{\pi}_v}_{t} \in \ca{L}_G 
(\bm{\rho}^{\bm{\pi}_v}_{t})$.
We remark that if $t \in [t_i, t_{i+1})$, the graph $\wt{G}_i$ 
introduced in Section \ref{sec:heat-equation} is obtained by  contracting $G_{v,t}$. 
%We fix $T\geq 0$. 
%Then, $g_v(T)$ can be written by a form of a Rayleigh quotient: 
%\[
% g_v(T) = 2\frac{\langle 
% D^{-1/2}(\bm{\rho}_T^{\bm{\pi}_v}-\pi), D^{-1/2}\ca{L}_{G_v,T}D^{1/2}(D^{-1/2}(\bm{\rho}_T^{\bm{\pi}_v}-\pi)\rangle_{I}}
% {\| D^{-1/2}(\bm{\rho}_T^{\bm{\pi}_v} - \bm{\pi}) \|^2_I}.
%\]
We fix $T \geq 0$ and consider a small $\Delta>0$. Let 
$t = T + \Delta$. 
Then, 
$\bm{\rho}^{\bm{\pi}_v}_{t}$ can be written by 
\[
 \bm{\rho}^{\bm{\pi}_v}_{t} = 
 \sum_{j = 1}^{n} a_j
 e^{-\lambda_j(G_{v,T})(T+\Delta)}\bm{u}_j(G_{v,T}),
\]
for some $a_j \in \RR$. Here, 
$0=\lambda_1(G_{v,T})\leq \lambda_2(G_{v,T})\leq \cdots \leq \lambda_{ {n}}(G_{v,T}) \leq 2$ are the eigenvalues of
$\ca{L}_{G_{v,T}}$ and $\bm{u}_1(G_{v,T}), \dots, \bm{u}_{ {n}}(G_{v,T})$ are these eigenvectors such that $\{ D^{-1/2}\bm{u}_1(G_{v,T}), \dots, D^{-1/2}\bm{u}_{ {n}}(G_{v,T})\}$ is orthonormal.  
We set $j_0$ as
\[
 j_0 = \min\{ j \mid j\geq 2, a_j \neq 0\}.
\]
Then, the $g_v(T)$ can be rephrased as
\begin{align*}
 g_v(T) &= 2 \frac{\sum_{j = j_0}^{n} a_j^2
 \lambda_j(G_{v,T})e^{-2\lambda_j(G_{v,T})T}}
 {\sum_{j = j_0}^{n} a_j^2
 e^{-2\lambda_j(G_{v,T})T}}  \\
 &= 2 \lambda_{j_0}(G_{v,T})
 \left(\frac{1 + {(a_{j_0 + 1}/a_{j_0})}^2
 (\lambda_{j_0 + 1}(G_{v,T})/\lambda_{j_0}(
 G_{v,T}))e^{-2(\lambda_{j_0+1}(G_{v,T})-\lambda_{j_0}(G_{v,T}))T} + \cdots}
 {1 + {(a_{j_0 + 1}/a_{j_0})}^2
 e^{-2(\lambda_{j_0+1}(G_{v,T})-\lambda_{j_0}(G_{v,T}))T} + \cdots} \right).
\end{align*}
We define $h_v(T)$ so that $g_v(T) = 2 \lambda_{j_0}(G_{v,T}) h_v(T)$. Then, $h_v(T)$ goes to $1$ as $T$ increases. Hence,
$h_v(T)$ is close to $1$ for large $T$. If $j_0 = 2$ (equivalent to 
$\langle \bm{u}_2(G_{v,T}), \bm{\rho}^{\bm{\pi}_v}_{T} \rangle_{D^{-1}} \neq 0$), we can find a nearly Cheeger cut: 
\begin{cor}\label{cor:main-S-cor}
 Notation is the same as above. We assume that $\langle \bm{u}_2(G_{v,T}), \bm{\rho}^{\bm{\pi}_v}_{T} \rangle_{D^{-1}} \neq 0$. Then,
 we have the following inequality:
\[
  4\phi_G h_v(T) \geq {(\wt{\kappa}^v_{T,t})}^2.
\]
\end{cor}
\begin{proof}
 By the assumption $j_0 = 2$, Theorem \ref{thm:main-S}, and Cheeger's inequality for graphs \eqref{eq:cheeger-graph}, we have 
 \[
 g_v(T) = 2\lambda_2(G_{v,T}) h_v(T) \leq 4\phi_{G_{v,T}} h_v(T).  
 \]
It is easy to see that $\phi_{G_{v,T}}(S) \leq \phi_G(S)$ holds for any $S \subset V$. This completes the proof.  
\end{proof}

To deduce Theorem~\ref{thm:main-intro} and Corollary~\ref{cor:main-intro-cor}, we need to show a relation 
$\wt{\kappa}^v_{T,t}$ with $\kappa^v_{T,t}$. 
The following relates the conductance of a sweep set in a hypergraph $G$ and that in a graph $\wt{G}$.
\begin{lem}\label{lem:sweep}
  Let $G=(V,E,w)$ be a hypergraph, $\bm{x} \in \mathbb{R}^V$ be a vector, $a$ be a real number, and
  %Let $G=(V,E,w)$ be a hypergraph, $\bm{x} \in \mathbb{R}^V$ be a vector, $a\in \RR$, and
  $\sigma^\ast$ be the ordered equivalence relation compatible with $\bm{x}$ in the sense of~\cite[Section~3]{chan2017diffusion},
  i.e., $u$ and $v$ are $\sigma^\ast$-equivalent if and only if $\bm{x}(u) = \bm{x}(v)$.
%  For any $\emptyset \subsetneq S \subsetneq V$, we have $\phi_{G'}(S) \leq \phi_{G}(S)$ for every $G' = (V,E',w') \in \mathcal{G}(G,\bm{x})$.
Let $\wt{G} = (\wt{V}, \wt{E}, \wt{w})$ be the weighted graph defined as in Section~\ref{sec:heat-equation} with this equivalent relation $\sigma^\ast$.
 If $S^a \subseteq V$ (resp., $\wt{S}^a \subseteq \wt{V}$) is the sweep set on $G$ (resp., $\wt{G}$)  with $\bm{x}(u)\geq a$, then
 $\phi_G(S^a) = \phi_{\wt{G}}(\wt{S}^a)$ holds.
\end{lem}

\begin{proof}[Proof of Theorem~\ref{thm:main-intro} and 
Corollary~\ref{cor:main-intro-cor}]
  As $\wt{\kappa}_{T,t}^v = \kappa_{T,t}^v$ by Lemma~\ref{lem:sweep}, we see that
  %As $\wt{\kappa}_{t,S} = \kappa_{t}$ by Lemma~\ref{lem:sweep}, we have
  Theorem~\ref{thm:main-S} and Corollary~\ref{cor:main-S-cor} imply Theorem~\ref{thm:main-intro} and 
  Corollary~\ref{cor:main-intro-cor}, respectively. 
\end{proof}

\section{Existence and Uniqueness of Solution}\label{sec:solution}
% In this section, we show the existence of a solution of the following Heat equation~\eqref{eq:heat-equation-hypergraph}.
% % : for $\bm{x} \in \RR^V$,
% % \begin{eqnarray*}
% % (\mathrm{HE};\bm{x})  \ \ \ \
% % \begin{cases}
% %  \ds\frac{d}{dt} \bm{u}(t) \in -L_G \ \! \bm{u}(t) \ \ (t\in(0,T])  \\
% %  \bm{u}(0) = \bm{x}.
% %  \end{cases}
% % \end{eqnarray*}

% \subsection{Existence of solution}
In this section, we show the existence and uniqueness of a solution to the heat equation~\eqref{eq:heat-equation-hypergraph} using the theory of monotone operators.
We refer the interested reader to the books by Miyadera~\cite{miyadera1992nonlinear} and Showalter~\cite{showalter2013monotone} for a detailed description of this topic.

We begin by introducing some definitions.
Let $X = (X, \langle \cdot,\cdot\rangle)$ be a Hilbert space, $\| \cdot \|$ be the norm defined from the inner product, and $A\colon X \to 2^X$ be a multi-valued operator on $X$.
Let $R(A) \subseteq X$ be the \emph{range} of $A$.
We often identify $A$ with the \emph{graph} of $A$; that is, $\{(x,y) \mid x \in X, y \in A(x) \} \subseteq X \times X$.

\begin{df}
  An operator $A\colon X \to 2^X$ is \emph{monotone} (or \emph{accretive}) if, for any
  $x,x'\in X$ and $y\in A(x), y'\in A(x')$, we have
  \[
    \langle y-y', x-x' \rangle \geq 0.
  \]
  When $-A$ is monotone, $A$ is called \emph{dissipative}.
\end{df}

\begin{df}
  A monotone operator $A\colon X \to 2^X$ is \emph{maximal} if $A$ is maximal as a graph of the monotone operator on $X$; i.e., if there is a monotone operator $B\colon X \to 2^X$ with $A(x) \subseteq B(x)$ for any $x \in X$. Then we have $A=B$.
\end{df}

% In this section, we identify the normalized Laplacian $L_G\colon \mathbb{R}^V \to 2^{\mathbb{R}^V}$ with a set $S \subseteq \RR^V \times \RR^V$ such that $S = \{(\bm{x},\bm{y}) \mid \bm{x} \in \mathbb{R}^V, \bm{y} \in L_G(\bm{x})  \}$.
To show that the heat equation~\eqref{eq:heat-equation-hypergraph} has a unique global solution, by the theory of monotone operators, it is sufficient to show that $\mathcal{L}_G\colon  \RR^V \to 2^{\RR^V}$ is a maximal monotone operator.
In our case, the Hilbert space is $X = \mathbb{R}^V$ equipped with the inner product $\inprod<\cdot,\cdot>_{D^{-1}}$ for $\bm{x}, \bm{y} \in \RR^V$.
%, because we regard this space as the measure space in the sense of~\cite{Chan:2018eu}.

\begin{lem}\label{lem:monotone}
  The operator $\mathcal{L}_G$ is monotone.
\end{lem}
 \begin{proof} % [Proof of Lemma~\ref{lem:monotone}]
  %\todo{多分、$D$ での weight 調整が必要. Laplacian を対称化するために.}
  %\ynote{can't compile Japanese text}
  For any $\bm{x}\in \mathbb{R}^V$ and $\bm{y} \in \mathcal{L}_G(\bm{x})$, we can write
  \[
    \bm{y} = B W B^\top D^{-1} \bm{x} = \sum_{e \in E}w(e) \bm{b}_e \bm{b}_e^\top \overline{\bm{x}},
  \]
  where $\overline{\bm{x}} = D^{-1}\bm{x}$. Further, $W \in \mathbb{R}^{E \times E}$ is a diagonal matrix with the $(e,e)$-th entry being $w(e)$. $B = {(\bm{b}_e)}_{e\in E}$ is a matrix with column vectors $\bm{b}_e \in \RR^V$, for which
  \[
    \bm{b}_e \in \argmax_{\bm{b} \in \bm{b}_e} \langle \bm{b}, \overline{\bm{x}} \rangle.
  \]
%  and $B(F_e)$ is the base polytope with respect to submodular function $F_e$
%  corresponding to hyperedge $e$ as in \S.2.1.
  We use this to show monotonicity.
  For $\bm{x}_1, \bm{x}_2 \in \mathbb{R}^V$ and $\bm{y}_1 \in \mathcal{L}_G(\bm{x}_1), \bm{y}_2 \in \mathcal{L}_G(\bm{x}_2)$, we have
  \[
    \bm{y}_1 = B_1 W B_1^\top \overline{\bm{x}}_1, \quad \bm{y}_2 = B_2 W B_2^\top \overline{\bm{x}}_2.
  \]
  Then, we have
  \begin{align*}
    \inprod<\bm{y}_1-\bm{y}_2, \bm{x}_1-\bm{x}_2>_{D^{-1}}
    &= \inprod<\bm{y}_1,\bm{x}_1>_{D^{-1}} +
    \inprod<\bm{y}_2,\bm{x}_2>_{D^{-1}} -\inprod<\bm{y}_2,\bm{x}_1>_{D^{-1}} - \inprod<\bm{y}_1,\bm{x}_2>_{D^{-1}}  \\
    &= \| B_1^\top \overline{\bm{x}}_1 \|_W^2 + \| B_2^\top \overline{\bm{x}}_2\|_W^2 - \overline{\bm{x}}_2^\top B_2 W B_2^\top \overline{\bm{x}}_1 - \overline{\bm{x}}_1^\top B_1 W B_1^\top \overline{\bm{x}}_2 \\
    &\geq \| B_1^\top \overline{\bm{x}}_1 \|_W^2 + \| B_2^\top \overline{\bm{x}}_2\|_W^2 - \overline{\bm{x}}_2^\top B_2 W B_1^\top \overline{\bm{x}}_1 - \overline{\bm{x}}_1^\top B_1 W B_2^\top \overline{\bm{x}}_2 \\
    & =  \| B_1^\top \overline{\bm{x}}_1 - B_2^\top \overline{\bm{x}}_2 \|_W^2
    \geq 0. \qedhere
  \end{align*}
\end{proof}

\begin{lem}\label{lem:maximal}
  The operator $\mathcal{L}_G$ is maximal.
\end{lem}
\begin{proof} % [Proof of Lemma~\ref{lem:maximal}]
 By~\cite[IV.1. Proposition 1.6]{showalter2013monotone}, it is sufficient to show that
 $R(I+ \mathcal{L}_G) = \RR^V$.
 This condition means that, for any $\bm{b} \in \RR^V$,
  the equation $\bm{x}+ \mathcal{L}_G(\bm{x}) \ni \bm{b}$ has a solution $\bm{x}$ in $\RR^V$.
  In a previous work~\cite[Section~3.1]{fujii2018polynomial}, an equivalent condition of the existence of the solution to $\mathcal{L}_G(\bm{x}) \ni \bm{b}$ was given.
  By a similar argument, we can give an equivalent condition for $\bm{x} + \mathcal{L}_G(\bm{x}) \ni \bm{b}$ and show the existence of the solution to $\bm{x}+\mathcal{L}_G(\bm{x}) \ni\bm{b}$.
\end{proof}

We obtain the following corollary using the theory of nonlinear semigroup:
\begin{cor}\label{cor:unique-solution}
  The heat equation~\eqref{eq:heat-equation-hypergraph} has a unique global solution.
\end{cor}
\begin{proof}
  Immediate from Lemmas~\ref{lem:monotone} and~\ref{lem:maximal}. See~\cite[IV, Proposition 3.1]{showalter2013monotone} for details.
\end{proof}

%\begin{proof}
%  This proof is an immediate consequence of Lemmas~\ref{lem:monotone} and~\ref{lem:maximal}. See~\cite[IV, Proposition 3.1]{showalter2013monotone} for details.
%\end{proof}

%\input{triangle-inequality}
%!TEX root=./main.tex

\newcommand{\prox}{\mathop{\mathrm{prox}}}

\section{Computation and Error Analysis of Difference Approximation}\label{sec:difference-approximation}

In this section, we prove Theorem~\ref{thm:difference-approximation-intro}.
In what follows, we fix a hypergraph $G=(V,E,w)$, $v \in V$, $T \geq 1$, and $\lambda \in (0,1)$.

We first review the construction of difference approximation $\bm{\rho}^\lambda_t$ given in~\cite[Section 5.3]{miyadera1992nonlinear}.
By the condition~(5.27) in~\cite{miyadera1992nonlinear} and the maximality of $\mathcal{L}_G$, for any $\bm{x} \in \RR^V$, there is a real number $\mu$ satisfying the following conditions:
\begin{equation}\label{eq:one-step}
  \begin{cases}
    0 < \mu \leq \lambda, \\
    \bm{x}_\mu \in \RR^V, \
    \bm{y}_\mu \in -\mathcal{L}_G (\bm{x}_\mu), \\
    \| \bm{x}_\mu - \bm{x} - \mu \bm{y}_\mu \|_{D^{-1}} < \mu \lambda.
  \end{cases}
\end{equation}

We define $\mu(\bm{x})$ as the least upper bound on $\mu$ satisfying~\eqref{eq:one-step}.
We consider an initial vector $\bm{x}_0\in\RR^V$.
Then, there is $h_1\in \RR$ such that $\mu(\bm{x}_0)/2 < h_1 \leq \lambda$ and there are $\bm{x}_1 \in \RR^V$ and $\bm{y}_1 \in -\mathcal{L}_G (\bm{x}_1)$ satisfying $\| \bm{x}_1 - \bm{x}_0 -h_1 \bm{y}_1 \|_{D^{-1}} < h_1\lambda$.
By repeating this argument, we can take sequences $\{ h_k \}$, $\{\bm{x}_k\}$, and $\{\bm{y}_k \}$ for $k = 1,2, \dots$ satisfying the following conditions:
\begin{enumerate}
  \itemsep=0pt
  \item$\mu(\bm{x}_{k-1})/2 < h_k \leq \lambda$,
  \item $\| \bm{x}_k - \bm{x}_{k-1} - h_k \bm{y}_k \|_{D^{-1}} < h_k \lambda$.
\end{enumerate}
Let $t_k = \sum_{j=1}^k h_j$.
Then, it is easy to show that $\{ t_k \}$, $\{ \bm{x}_k\}$, and $\{ \bm{y}_k \}$ satisfy the following conditions for $\{ t^\lambda_k \}$, $\{ \bm{x}^\lambda_k\}$, and $\{ \bm{y}^\lambda_k \}$: 
%in Lemma~\ref{lem:difference-approximation}. 
 \begin{enumerate}
  \itemsep=0pt
  \item $0= t_0^\lambda < t_1^\lambda < \cdots < t^\lambda_k < \cdots$ with $\lim_{k\to \infty} t^\lambda_k = \infty$,
  \item $t^\lambda_k - t_{k-1}^\lambda <\lambda \quad (k = 1,2, \dots)$,
  \item $\| \bm{x}^\lambda_k - \bm{x}^\lambda_{k-1} - (t^\lambda_k - t_{k-1}^\lambda)\bm{y}^\lambda_k \|_{D^{-1}} <
    \lambda (t^\lambda_k - t_{k-1}^\lambda) \quad (k= 1,2,\dots)$. 
    %\label{item:difference-approximation-3}
  \end{enumerate}
Then, the function $\bm{\rho}^\lambda_t$ was defined by
\begin{equation}\label{eq:difference-approximation}
\bm{\rho}^\lambda_t =
 \begin{cases}
 \bm{x}_0 & \mathrm{if} \ t = 0, \\
 \bm{x}^\lambda_k & \mathrm{if} \ t \in (t^\lambda_k, t^\lambda_{k+1}]\cap (0,T].
 \end{cases}
\end{equation}
Theorem~\ref{thm:difference-approximation-intro} follows from Lemmas~\ref{lem:difference-approximation-computation} and~\ref{lem:difference-approximation-error} below.

\begin{lem}\label{lem:difference-approximation-computation}
  We can compute (a concise representation) of ${\{\bm{\rho}_t^\lambda\}}_{0 \leq t \leq T}$ for every $0 \leq t \leq T$ in time polynomial in $1/\lambda$, $T$, and $\sum_{e \in E}|e|$.
\end{lem}
\begin{proof} % [Proof of Lemma~\ref{lem:difference-approximation-computation}]
  From the construction of $\bm{\rho}_t^\lambda$, it suffices to compute $\bm{x}^\lambda_k$ until $t_k \geq T$.
  Note that we can obtain $\bm{x}^\lambda_k$ from $\bm{x}^\lambda_{k-1}$ by solving the equation
  \begin{align}
    \bm{x} - \bm{x}^\lambda_{k-1} \in -\lambda \mathcal{L}_G (\bm{x}), \label{eq:equation-in-one-step}
  \end{align}
  because, then, we can set $h_k = \lambda$ and $\bm{x}^\lambda_k$ to be the obtained solution.

Let $\ol{\bm{x}} = D^{-1}\bm{x}$ for any $\bm{x}\in \RR^V$.
Then, solving~\eqref{eq:equation-in-one-step} is equivalent to solving
\begin{equation}\label{eq:normalized-equation}
 D\ol{\bm{x}} - D\ol{\bm{x}}^\lambda_{k-1} \in -\lambda L_G (\ol{\bm{x}}).
\end{equation}
  By an argument similar to~\cite[Section 3.1]{fujii2018polynomial}, solving~\eqref{eq:normalized-equation} is equivalent to computing the following proximal operator
  \begin{equation}\label{eq:minimization}
    \prox(\ol{\bm{x}}^\lambda_{k-1}) := \argmin_{\ol{\bm{x}} \in \RR^V} \left( \frac{\lambda}{2} \sum_{e\in E}w(e) {f_e(\ol{\bm{x}})}^2 +
    \frac{1}{2} \|\ol{\bm{x}} - \ol{\bm{x}}^\lambda_{k-1}\|_{D}^2 \right),
  \end{equation}
  which can be computed in time polynomial in $\sum_{e \in E}|\mathcal{V}_e|$, where $\mathcal{V}_e$ is the set of extreme points of $\bm{b}_e$~\cite[Theorem D.1 (i)]{fujii2018polynomial}.
  As $\mathcal{V}_e \leq |e|^2$, we can compute $\bm{x}^\lambda_k=D\prox(\ol{\bm{x}}^\lambda_{k-1})$ in time polynomial in $\sum_{e \in E}|e|$.

  As $h_k = \lambda$, we need to compute $\bm{x}^\lambda_k$ for $k \leq \lceil T/k \rceil$.  Hence, the total time complexity is polynomial in $1/\lambda$, $T$, and $\sum_{e \in E}|e|$.
\end{proof}

\begin{lem}\label{lem:difference-approximation-error}
  % Let $\bm{\rho}_t$ be a solution to $(\mathrm{HE};\bm{x})$ and $\bm{\rho}^\lambda_t$ be its difference approximation defined in~\eqref{eq:difference-approximation} with an initial vector $\bm{x}$.
  % Then, for any $t \in [0,T]$ and any $\lambda > 0$, we have
  We have $\| \bm{\rho}^\lambda_t - \bm{\rho}^{\bm{\pi}_v}_t \|_{D^{-1}} = O(\sqrt{\lambda T})$.
\end{lem}
\begin{proof} % [Proof of Lemma~\ref{lem:difference-approximation-error}]
  Let $|||\mathcal{L}_G(\bm{x}) ||| = \inf \{ \| \bm{y}\|_{D^{-1}} \mid \bm{y} \in \mathcal{L}_G(\bm{x}) \}$.
  We set $N_\lambda \in \ZZ_+$ as $t_{N_\lambda}^\lambda < T\leq t_{N_\lambda}^\lambda$, $| \Delta_\lambda| = \max \{ t^\lambda_k - t^\lambda_{k-1}; k= 1,2, \dots, N_\lambda \}$ and $\mathcal{E}_\lambda = \sum_{k=1}^{N_\lambda} \| \bm{\mathcal{E}}_k^\lambda \|_{D^{-1}}(t^\lambda_k - t^\lambda_{k-1})$, where $\bm{\mathcal{E}}^\lambda_k$ is defined as
  \[
    \bm{\mathcal{E}}^\lambda_k = \frac{\bm{x}^\lambda_k - \bm{x}^\lambda_{k-1}}{t^\lambda_k - t^\lambda_{k-1}} - \bm{y}^\lambda_k \ \ \ (k = 1, 2, \dots).
  \]
  Then, by the equation (5.20) of~\cite{miyadera1992nonlinear} instantiated with $\omega_0 = 0$, $t= s$, $x_p = x$, we have
  \[
   \| \bm{\rho}^\lambda_t - \bm{\rho}^\mu_t \|_{D^{-1}} \leq  \mathcal{E}_\lambda + \mathcal{E}_\mu + {\left({(|\Delta_\lambda| + |\Delta_\mu|)}^2 + | \Delta_\lambda| (t + | \Delta_\lambda | ) +
    | \Delta_\mu | (t + | \Delta_\mu | )  \right)}^\frac{1}{2} \times ||| \mathcal{L}_G(\bm{\pi}_v) |||
  \]
  for $t \in [0,T]$ and $\mu > 0$.
  The condition 3 for $\{t_k^\lambda\}, \{\bm{x}_k^\lambda\}$, and 
  $\{\bm{y}_k^\lambda\}$ implies $\| \bm{\mathcal{E}}^\lambda_k \|_{D^{-1}} < \lambda$. Hence,  $\mathcal{E}_\lambda < \lambda t_{N_\lambda}^\lambda < \lambda (T + \lambda)$ as  $t^\lambda_{N_\lambda-1} < T \leq  t^\lambda_{N_\lambda}$.

  Therefore by taking limit $\mu \to 0+$, we have
  \[
    \| \bm{\rho}^\lambda_t - \bm{\rho}_t^{\bm{\pi}_v} \|_{D^{-1}}
    < \lambda(T+\lambda) + \sqrt{\lambda^2 + \lambda(t+ \lambda)} ||| \mathcal{L}_G (\bm{\pi}_v) |||
    = O(\sqrt{\lambda T}). \qedhere
  %  < \lambda T + \left(\lambda + \sqrt{\lambda T} \right) ||| \mathcal{L}_G \bm{x} ||| \\
  %  & <
  %  < \lambda(T+\lambda) + \left(\lambda + \sqrt{\lambda(T+ \lambda)} \right) ||| \mathcal{L}_G \bm{x} ||| \\
  %  = \sqrt{\lambda} [\sqrt{\lambda}(T + \lambda) + \sqrt{\lambda + (t+ \lambda)} ||| \mathcal{L}_G(\bm{x})|||]
  %  = O(\sqrt{\lambda}).
  \]
\end{proof}

\appendix

%!TEX root=main.tex

\section{Proof of Theorem~\ref{thm:graph-diffusion}}\label{sec:proof-of-graph-diffusion}

\begin{proof}%[Proof of Theorem~\ref{thm:graph-diffusion}]
By Corollary~\ref{cor:unique-solution},  for any initial vector $\bm{s}$, there exists a unique solution $\bm{\rho}_{t}^{\bm{s}}$ of~\eqref{eq:heat-equation-hypergraph}.
Let $\bm{\mu}^{\bm{s}}_{t}=D^{-1} \bm{\rho}^{\bm{s}}_t$.
By~\cite[\S.3 and \S.4]{chan2017diffusion}, we can compute
any higher right derivatives $\frac{d^{n}\bm{\mu}_{t}^{\bm{s}}}{dt^n}|_{t=0}$.
Let $(\sigma^\ast, \succ)$ be the lexicographical ordered equivalence relation on $V$
consistent with ${\{ d^{n}\bm{\mu}_{t}^{\bm{s}}/dt^n |_{t=0} \}}_n$.

%as follows: for $u,v \in V$, $u$ and $v$ are $\sigma^\ast$-equivalence (denoted by $u \sim_{\sigma^\ast} \! v$) if for any $n\in \ZZ_+$,
%$\frac{d^{n}\bm{\rho}_{t}^{\bm{s}}}{dt^n}|_{t=0}(u) = \frac{d^{n}\bm{\rho}_{t}^{\bm{s}}}{dt^n}|_{t=0}(v)$ holds. Let $\mathcal{C}[\sigma^\ast] = \{U_1, \dots, U_m \} \subseteq 2^V$ be the family of $\sigma^\ast$-equivalent
%classes. Then, we define the lexicographical order $\succ$ as follows: $U_k \succ U_l$ if there is an integer $n \in \ZZ_+$ such that, for $u_k \in U_k$ and $u_l \in U_l$, the followings hold:
%\begin{align*}
%\left.\frac{d^{k}\bm{\rho}_{t}^{\bm{s}}}{dt^k}\right|_{t=0}(u_k) &= \left.\frac{d^{k}\bm{\rho}_{t}^{\bm{s}}}{dt^k}\right|_{t=0}(u_l) \ \ \ \mathrm{for} \ k=0, \dots, m-1, \\
% \left.\frac{d^{m}\bm{\rho}_{t}^{\bm{s}}}{dt^m}\right|_{t=0}(u_k) & > \left.\frac{d^{m}\bm{\rho}_{t}^{\bm{s}}}{dt^m}\right|_{t=0}(u_l).
%\end{align*}
%We define $\succeq$ as $\succ$ or $=$.

For each $e \in E$, let $S_e$, $I_e$, $S_e^{\sigma^\ast}$, and $I_e^{\sigma^\ast}$ be subsets introduced in \S.2.1 and \S.2.2. Let $G' = (V', E', w')$ be the undirected
weighted graph with respect to $\bm{\mu}^{\bm{s}}_t$ as in \S.\ref{subsec:laplacian-for-hypergraphs}.
We remark that for each $e\in E$,  $w'_e(uv) \neq 0$ only if $(u,v)$ or $(v,u)$ is in $S_e^{\sigma^\ast} \times I_e^{\sigma^\ast}$, because any vertex in $S_e\!\setminus \!S_e^{\sigma^\ast}$ (resp.
$I_e\!\setminus \! I_e^{\sigma^\ast}$) will leave
$S_e$ (resp. $I_e$) after infinitesimal time.
We take $T>0$ such that if we retake ordered equivalence relation
$(\sigma^\ast,\succ)$ consistent with
$\{d^{n}\bm{\mu}_{t}^{\bm{s}}/dt^n |_{t=t'} \}_n$ at $t' \in [0, T]$, for $e\in E$,
$S_e^{\sigma^\ast}$ and $I_e^{\sigma^\ast}$ do not change.
 Then, for any $t\in [0,T]$,  we have
\allowdisplaybreaks[1]
\begin{align*}
 -L_G (\bm{\rho}_{t}^{\bm{s}}) =
 - (I - A_G D^{-1}) &(\bm{\rho}_{t}^{\bm{s}}) \ni - \bm{\rho}_{t}^{\bm{s}}
  + A_{\bm{\mu}_{t}^{\bm{s}}} \bm{\mu}_{t}^{\bm{s}} \\
 & = -\bm{\rho}_{t}^{\bm{s}} + {\left( \sum_{v \in V} w'(uv) \bm{\mu}_{t}^{\bm{s}}(v)  \right)}_u \\
 & = {\left( - d_{G}(u) \bm{\mu}_{t}^{\bm{s}}(u) + \sum_{v \in V} w'(uv) \bm{\mu}_{t}^{\bm{s}}(v) \right)}_u \\
 & = {\left(- \left(\sum_{v \in V} w'(uv) \right)  \bm{\mu}_{t}^{\bm{s}}(u) + \sum_{v \in V} w'(uv) \bm{\mu}_{t}^{\bm{s}}(v) \right)}_u \\
 & = {\left(- \left(\sum_{v \in V, \atop v \neq u} w'(uv) \right)  \bm{\mu}_{t}^{\bm{s}}(u) +
  \sum_{v \in V, \atop v \neq u} w'(uv) \bm{\mu}_{t}^{\bm{s}}(v) \right)}_u \\
 & = {\left(- \left(\sum_{v \in V, \atop v \neq u} w'(uv) \right)  (\bm{\mu}_{t}^{\bm{s}}(u) - \bm{\mu}_{t}^{\bm{s}}(v)) \right)}_u \\
 & = {\left(- \sum_{e\in E }\left(\sum_{v \in V, \atop v \neq u} w'_e(uv) \right)  (\bm{\mu}_{t}^{\bm{s}}(u) - \bm{\mu}_{t}^{\bm{s}}(v)) \right)}_u.
\end{align*}
If $u \in S_e^{\sigma^\ast}$, then $w'_e(uv) \neq 0$ holds only if $v \in I_e^{\sigma^\ast}$. Hence, we have
\[
\left(\sum_{v \in V, \atop v \neq u} w'_e(uv) \right)  (\bm{\mu}_{t}^{\bm{s}}(u) - \bm{\mu}_{t}^{\bm{s}}(v)) =
\left(\sum_{v \in I_e^{\sigma^\ast}} w'_e(uv) \right)  \Delta_e(\bm{\mu}_{t}^{\bm{s}}),
\]
where $\Delta_e(\bm{\mu}_{t}^{\bm{s}}) =  \max_{u,v \in e}(\bm{\mu}_{t}^{\bm{s}}(u) - \bm{\mu}_{t}^{\bm{s}}(v))$.
On the other hand, if
$u \in I_e^{\sigma^\ast}$, then $w'_e(uv) \neq 0$ holds only if $v \in S_e^{\sigma^\ast}$. Hence, we have
\[
\left(\sum_{v \in V, \atop v \neq u} w'_e(uv) \right)  (\bm{\mu}_{t}^{\bm{s}}(u) - \bm{\mu}_{t}^{\bm{s}}(v)) =
- \left(\sum_{v \in S_e^{\sigma^\ast}} w'_e(uv) \right)  \Delta_e(\bm{\mu}_{t}^{\bm{s}}).
\]
Using these equalities, we obtain
\begin{align*}
 &{\left(- \sum_{e\in E }\left(\sum_{v \in V, \atop v \neq u} w'_e(uv) \right)  (\bm{\mu}_{t}^{\bm{s}}(u) - \bm{\mu}_{t}^{\bm{s}}(v)) \right)}_u \\
 &={\left(- \sum_{e\in E, \atop u\in S_e^{\sigma^\ast} }\left(\sum_{v \in I_e^{\sigma^\ast}} w'_e(uv) \right) \Delta_e(\bm{\mu}_{t}^{\bm{s}})
 + \sum_{e\in E, \atop u\in I_e^{\sigma^\ast} }\left(\sum_{v \in S_e^{\sigma^\ast}} w'_e(uv) \right) \Delta_e(\bm{\mu}_{t}^{\bm{s}}) \right)}_u.
\end{align*}

Let $\mathcal{C}[\sigma^\ast] = \{ U_1, \dots, U_m \}$ be the family of $\sigma^\ast$-equivalence classes such that
$U_k \succ U_{i+1}$ and we fix $u_k \in U_k$ for each $i$. For $u \in U_k$, we note ${[u]}_{\sigma^\ast} = U_k$.

We sum up the entries of the above vector along $U_k$. Then, we have
\begin{align*}
& \sum_{u\in U_k} \left(- \sum_{e\in E, \atop u\in S_e^{\sigma^\ast} }\left(\sum_{v \in I_e^{\sigma^\ast}} w'_e(uv) \right) \Delta_e(\bm{\mu}_{t}^{\bm{s}})
+ \sum_{e\in E, \atop u\in I_e^{\sigma^\ast} }\left(\sum_{v \in S_e^{\sigma^\ast}} w'_e(uv) \right) \Delta_e(\bm{\mu}_{t}^{\bm{s}}) \right) \\
&= - \sum_{e\in E, \atop S_e^{\sigma^\ast}\cap U_k \neq \emptyset }\left(\sum_{u \in S_e^{\sigma^\ast} \atop v \in I_e^{\sigma^\ast}} w'_e(uv) \right) \Delta_e(\bm{\mu}_{t}^{\bm{s}})
+ \sum_{e\in E, \atop I_e^{\sigma^\ast}\cap U_k \neq \emptyset }\left(\sum_{u \in I_e^{\sigma^\ast} \atop v \in S_e^{\sigma^\ast}} w'_e(uv) \right) \Delta_e(\bm{\mu}_{t}^{\bm{s}})  \\
&= - \sum_{e\in E, \atop S_e^{\sigma^\ast}\cap U_k \neq \emptyset }w_e \Delta_e(\bm{\mu}_{t}^{\bm{s}})
+ \sum_{e\in E, \atop I_e^{\sigma^\ast}\cap U_k \neq \emptyset }w_e \Delta_e(\bm{\mu}_{t}^{\bm{s}}).
\end{align*}
We remark that the last form is independent of the choice of $w'_e(uv)$.
Now, if $S_e^{\sigma^\ast}\cap U_k \neq \emptyset$, $\Delta_e(\bm{\mu}^{\bm{s}}_t) = \bm{\mu}_{t}^{\bm{s}}(u_k) - \bm{\mu}_{t}^{\bm{s}}(u_l)$ for some $l>k$ such that
$I_e^{\sigma^\ast} \cap U_l \neq \emptyset$.
Similarly, if $I_e^{\sigma^\ast}\cap U_k \neq \emptyset$, $\Delta_e(\bm{\mu}^{\bm{s}}_t) = \bm{\mu}_{t}^{\bm{s}}(u_l) - \bm{\mu}_{t}^{\bm{s}}(u_k)$ for some $l<k$ such that
$S_e^{\sigma^\ast} \cap U_l \neq \emptyset$.
 Hence the sum becomes
\begin{align*}
& - \sum_{e\in E, \atop S_e^{\sigma^\ast}\cap U_k \neq \emptyset }w_e \Delta_e(\bm{\mu}_{t}^{\bm{s}})
+ \sum_{e\in E, \atop I_e^{\sigma^\ast}\cap U_k \neq \emptyset }w_e \Delta_e(\bm{\mu}_{t}^{\bm{s}}) \\
& = - \sum_{e\in E, \atop S_e^{\sigma^\ast}\cap U_k \neq \emptyset }
\sum_{ l \neq k \atop  I_e^{\sigma^\ast} \cap U_l \neq \emptyset} w_e (\bm{\mu}_{t}^{\bm{s}}(u_k) - \bm{\mu}_{t}^{\bm{s}}(u_l))
+ \sum_{e\in E, \atop I_e^{\sigma^\ast}\cap U_k \neq \emptyset }
\sum_{l \neq k \atop  S_e^{\sigma^\ast} \cap U_l \neq \emptyset} w_e (\bm{\mu}_{t}^{\bm{s}}(u_l) - \bm{\mu}_{t}^{\bm{s}}(u_k)) \\
& = - \sum_{ l \neq k} \sum_{e\in E, S_e^{\sigma^\ast}\cap U_k \neq \emptyset
\atop  I_e^{\sigma^\ast} \cap U_l \neq \emptyset} w_e (\bm{\mu}_{t}^{\bm{s}}(u_k) - \bm{\mu}_{t}^{\bm{s}}(u_l))
- \sum_{ l \neq k} \sum_{e\in E, I_e^{\sigma^\ast}\cap U_k \neq \emptyset
\atop  S_e^{\sigma^\ast} \cap U_l \neq \emptyset} w_e (\bm{\mu}_{t}^{\bm{s}}(u_k) - \bm{\mu}_{t}^{\bm{s}}(u_l)) \\
&= - \sum_{ l \neq k} a_{kl} (\bm{\mu}_{t}^{\bm{s}}(u_k) - \bm{\mu}_{t}^{\bm{s}}(u_l))
- \sum_{ l \neq k} a_{lk} (\bm{\mu}_{t}^{\bm{s}}(u_k) - \bm{\mu}_{t}^{\bm{s}}(u_l)), \\
& = - \left(\sum_{l \neq k} (a_{kl} + a_{lk}) \right) \bm{\mu}_{t}^{\bm{s}}(u_k) +  \sum_{l \neq k} (a_{kl} + a_{lk})  \bm{\mu}_{t}^{\bm{s}}(u_l),
%\\ & = - \left(\sum_{l \neq k} \wt{w}(u_ku_l) \right) \bm{\mu}_{t}^{\bm{s}}(u_k) +  \sum_{l \neq k} \wt{w}(u_ku_l) \bm{\mu}_{t}^{\bm{s}}(u_l)
\end{align*}
where
\begin{align*}
 a_{kl}  &= \sum_{e\in E, S_e^{\sigma^\ast}\cap U_k \neq \emptyset
\atop  I_e^{\sigma^\ast} \cap U_l \neq \emptyset} w_e.
\end{align*}
We set
\begin{align*}
\wt{w}(u_ku_l) &= a_{kl} + a_{lk} \ \ \mathrm{for \ } k \neq l
\quad \text{and}\quad
\wt{w}(u_ku_k) = d_{\tilde{G}}(u_k) - \sum_{l \neq k} \wt{w}(u_ku_l),
\end{align*}
where $d_{\tilde{G}}(u_k) = \sum_{i \in U_k} d_{G}(u)$.
For $\bm{\rho}_{t}^{\bm{s}}$, we set $\wt{\bm{\rho}}_{t}^{\bm{s}} = {\left( \sum_{u \in U_k} \bm{\rho}_{t}^{\bm{s}}(u)\right)}_k \in \RR^m$.
Then, we have
\[
\wt{\bm{\rho}}_{t}^{\bm{s}} = {\left( \left(\sum_{u \in U_k}d_{G}(u)\right)\bm{\mu}_{t}^{\bm{s}}(u_k)\right)}_k = {(d_{\tilde{G}}(u_k) \bm{\mu}_{t}^{\bm{s}}(u_k))}_k.
\]
Let $D_{\tilde{G}} = \mathrm{diag}(d_{\tilde{G}}(u_k))$ and $\wt{\bm{\mu}_{t}^{\bm{s}}}:= {(\bm{\mu}_{t}^{\bm{s}}(u_k))}_k =
D_{\tilde{G}}^{-1} \wt{\bm{\rho}}_{t}^{\bm{s}} \in \RR^m$.
Then, we have
\begin{align*}
&{\left( \sum_{u \in U_k} (- \bm{\rho}_{t}^{\bm{s}}(u) + (A_{\bm{\mu}_{t}^{\bm{s}}} \bm{\mu}_{t}^{\bm{s}})(u)) \right)}_k\\
& = {\left( - \left(\sum_{l \neq k} \wt{w}(u_ku_l)\right){\bm{\mu}_{t}^{\bm{s}}}(u_k) + \sum_{l \neq k} \wt{w}(u_ku_l)
{\bm{\mu}_{t}^{\bm{s}}}(u_l) \right)}_k \\
& = - \wt{\bm{\rho}}_{t}^{\bm{s}} + \wt{\bm{\rho}}_{t}^{\bm{s}} + {\left( - \left(\sum_{l \neq k} \wt{w}(u_ku_l)\right){\bm{\mu}_{t}^{\bm{s}}}(u_k) + \sum_{l \neq k} \wt{w}(u_ku_l){\bm{\mu}_{t}^{\bm{s}}}(u_l) \right)}_k \\
& = - \wt{\bm{\rho}}_{t}^{\bm{s}} + {\left( \left(d_{\tilde{G}}(u_k) - \sum_{l \neq k} \wt{w}(u_ku_l) \right) {\bm{\mu}_{t}^{\bm{s}}}(u_k) + \sum_{l \neq k} \wt{w}(u_ku_l)
{\bm{\mu}_{t}^{\bm{s}}}(u_l)\right)}_k \\
& = - \wt{\bm{\rho}}_{t}^{\bm{s}} + {(\wt{w}(u_ku_l))}_{k,l} \wt{\bm{\mu}_{t}^{\bm{s}}}
 = - (I - {(\wt{w}(u_ku_l))}_{k,l} D_{\tilde{G}}^{-1}) \wt{\bm{\rho}}_{t}^{\bm{s}}.
\end{align*}
This $I - {(\wt{w}(u_ku_l))}_{k,l} D_{\tilde{G}}^{-1}$ is the 
normalized graph Laplacian $\ca{L}_{\tilde{G}}$ introduced in Section~\ref{sec:heat-equation}.

We return to the heat equation.
We consider the solution $\bm{\rho}_{t}^{\bm{s}}$ of heat equation~\eqref{eq:heat-equation-hypergraph}.
We set $\bm{\mu}_{t}^{\bm{s}} = D_{G}^{-1} \bm{\rho}_{t}^{\bm{s}}$.
%Then, as we mentioned above, we have an ordered equivalence relation $\sigma^\ast$ compatible with $\bm{\rho}_{t}^{\bm{s}}$ at $t=0$ and we divide $V$ as $\sigma^\ast$-equivalence classes $\{ U_1, \dots, U_m \}$.
By the definition of $\sigma^\ast$, $\bm{\mu}_{t}^{\bm{s}}(u) = \bm{\mu}_{t}^{\bm{s}}(v)$ if $u \sim_{\sigma^\ast} \!\! v$ until the next tie occurs.
We remark that $\bm{\rho}_{t}^{\bm{s}}$ until the next tie occurs is determined by
$\bm{\mu}_{t}^{\bm{s}}(u_k)$, $i=1, \dots, m$.
Also $\frac{d\bm{\mu}_{t}^{\bm{s}}(u)}{dt} = \frac{d\bm{\mu}_{t}^{\bm{s}}(v)}{dt}$ holds for such $u,v$ and $t$.
Hence, we have
\begin{align*}
 \sum_{u\in U_k} \frac{d\bm{\rho}_{t}^{\bm{s}}}{dt}(u) = d_{\tilde{G}}(u_k) \frac{d\bm{\mu}_{t}^{\bm{s}}}{dt}(u_k)
% \label{eq:sum-of-differential}
\end{align*}
Let $\wt{\bm{\rho}}_{t}^{\bm{s}} = {\left(\sum_{u \in U_k}\bm{\rho}_{t}^{\bm{s}} (u)\right)}_k \in \RR^m$, and
$\wt{\bm{\mu}}_{t}^{\bm{s}}:= D_{\tilde{G}}^{-1} \wt{\bm{\rho}}_{t}^{\bm{s}}$.
By the argument above, $\wt{\bm{\rho}}_{t}^{\bm{s}}$ is the unique solution of the heat equation
\begin{equation}
  \frac{d\wt{\bm{\rho}}_{t}}{dt} = - \ca{L}_{\wt{G}}\wt{\bm{\rho}}_{t}, \ \ \wt{\bm{\rho}}_{0} = \wt{\bm{s}}.
  \label{eq:graph-heat-equation}
\end{equation}
This solution $\wt{\bm{\rho}}_{t}^{\bm{s}}$ determines $\wt{\bm{\mu}}_{t}^{\bm{s}}$, and hence $\bm{\mu}^{\bm{s}}_t$. If $u\in U_k$, then
\[
 \bm{\rho}_{t}^{\bm{s}}(u) = d_{G}(u)\bm{\mu}_{t}^{\bm{s}}(u_k)
\]
holds. Hence, we can recover $\bm{\rho}^{\bm{s}}_t$ from the heat equation~\eqref{eq:graph-heat-equation}.
\end{proof}

\section{Proofs of Section~\ref{sec:analysis}}\label{sec:proofs-of-useful-lemmas}
%\section{Proofs of Section~\ref{subsec:useful-lemmas}}\label{sec:proofs-of-useful-lemmas}

\subsection{Useful lemmas}\label{subsec:useful-lemmas}
In this section, we derive several inequalities on $f_i$ that will be useful later.
Note that the proofs are deferred to Section~\ref{sec:proofs-of-useful-lemmas}.
We define $\mathcal{R}_i\colon \mathbb{R}^{\wt{V}_i} \to \mathbb{R}$ as
\begin{align}
  \mathcal{R}_i(\bm{x}) = \frac{\bm{x}^\top L_{\wt{G}_i} \bm{x}}{\|\bm{x}\|^2_{D_{\tilde{G}_i}}} = \frac{ \sum_{uv \in \wt{E}_i} {\left( \bm{x}(u) - \bm{x}(v) \right)}^2 \wt{w}_i(uv)}{\sum_{v \in \wt{V}_i} {\bm{x}(v)}^2 d_{\wt{G}_i}(v)}.
  \label{eq:Rayleigh-quotient}
\end{align}

\begin{lem}\label{lem:log-derivative}
  For any $i \in \mathbb{Z}_+$, we have
  \[
     \frac{d}{d \Delta} \log f_i(\Delta)
     = \frac{ \wt{\bm{\rho}}_{i,0}^\top D_{\tilde{G}_i}^{-1}\frac{d}{d \Delta} \wt{\bm{\rho}}_{i,\Delta}  }{\wt{\bm{\rho}}_{i,0}^\top D_{\tilde{G}_i}^{-1}(\wt{\bm{\rho}}_{i,\Delta}-\wt{\bm{\pi}}^i)}
     = - \mathcal{R}_i\left( \frac{\wt{\bm{\rho}}_{i,\Delta/2}}{d_{\wt{G}_i}}-\frac{1}{\vol(\wt{V}_i)} \right).
     %\frac{ \sum_{uv \in \wt{E}_i} {\left( \frac{\wt{\bm{\rho}}_{i,\Delta/2}(u)}{d(u)} -\frac{\wt{\bm{\rho}}_{i,\Delta/2}(v)}{d(v)} \right)}^2 w_i(uv)}{\sum_{v \in V} {\left( \frac{\wt{\bm{\rho}}_{i,\Delta/2}(v)}{d(v)} -\frac{1}{\vol(V) }\right)}^2 d(v)}.
  \]
\end{lem}
\begin{proof}%[Proof of Lemma~\ref{lem:log-derivative}]
We first prove the following lemma:
  \begin{cla}\label{cla:numerator}
    For any $i \in \mathbb{Z}_+$ and $\Delta \geq 0$, we have
    \begin{align*}
      \wt{\bm{\rho}}_{i,0}^\top D_{\tilde{G}_i}^{-1} \frac{d \wt{\bm{\rho}}_{i,\Delta}}{d \Delta}
      &= - {(D_{\tilde{G}_i}^{-1}\wt{\bm{\rho}}_{i,\Delta/2})}^\top (D_{\tilde{G}_i} - A_{\wt{G}_i}) (D_{\tilde{G}_i}^{-1} \wt{\bm{\rho}}_{i,\Delta/2})\\
      &=  - \sum_{ uv \in \wt{E}_i } {\left( \frac{\wt{\bm{\rho}}_{i,\Delta/2}(u)}{d_{\wt{G}_i}(u)} -\frac{\wt{\bm{\rho}}_{i,\Delta/2}(v)}{d_{\wt{G}_i}(v)} \right)}^2 \wt{w}_i(uv) \leq 0,
    \end{align*}
    where $A_{\wt{G}_i}$ is the adjacency matrix of $\wt{G}_i$.
  \end{cla}
  \begin{proof}
    We have
    \begin{align*}
      & \wt{\bm{\rho}}_{i,0}^\top D_{\tilde{G}_i}^{-1} \frac{d \wt{\bm{\rho}}_{i,\Delta}}{d \Delta}
      =  -\wt{\bm{\rho}}_{i,0}^\top D_{\tilde{G}_i}^{-1} H_{i,\Delta} \mathcal{L}_i \wt{\bm{\rho}}_{i,0}
      = - \wt{\bm{\rho}}_{i,0}^\top D_{\tilde{G}_i}^{-1} H_{i,\Delta/2}H_{i,\Delta/2} \mathcal{L}_i \wt{\bm{\rho}}_{i,0} \tag{by $H_{i,\Delta} = H_{i,\Delta/2}H_{i,\Delta/2}$} \\
      & = -\wt{\bm{\rho}}_{i,0}^\top {(H_{i,\Delta/2})}^\top  D_{\tilde{G}_i}^{-1} (D_{\tilde{G}_i} - A_{\wt{G}_i}) D_{\tilde{G}_i}^{-1} H_{i,\Delta/2} \wt{\bm{\rho}}_{i,0} \tag{by $D_{\tilde{G}_i} H_{i,\Delta/2} ={(H_{i,\Delta/2})}^\top D_{\tilde{G}_i}$}\\
      & = - {(D_{\tilde{G}_i}^{-1}\wt{\bm{\rho}}_{i,\Delta/2})}^\top (D_{\tilde{G}_i} - A_{\wt{G}_i}) (D_{\tilde{G}_i}^{-1} \wt{\bm{\rho}}_{i,\Delta/2}).
  %%    = - \sum_{uv} {\left( \frac{\wt{\bm{\rho}}_{i,\Delta/2}(u)}{d(u)} -\frac{\wt{\bm{\rho}}_{i,\Delta/2}(v)}{d(v)} \right)}^2 w_i(uv).
    \end{align*}
    The second equality in the statement is obtained through a direct calculation.
  %% and the third equality is by $$ as $\wt{D}^{-1}\mathcal{L}_i^n = {(\mathcal{L}_i^n)}^\top \wt{D}^{-1}$.
  \end{proof}

  We are now ready to prove Lemma~\ref{lem:log-derivative}. The first equality is obtained through direct calculation and the second equality follows from Proposition~\ref{prop:norm} and Lemma~\ref{cla:numerator}.
\end{proof}

\begin{lem}\label{lem:log-second-derivative}
  For any $i \in \mathbb{Z}_+$, we have
  \[
    \frac{d^2}{ d\Delta^2} \log f_i(\Delta) \geq 0.
  \]
\end{lem}
\begin{proof}%[Proof of Lemma~\ref{lem:log-second-derivative}]
  By Lemma~\ref{lem:log-derivative}, we have
  \begin{align*}
   & -\frac{d^2}{ d\Delta^2}(-\log(\wt{\bm{\rho}}_{i,0}^\top D_{\tilde{G}_i}^{-1} (\wt{\bm{\rho}}_{i,\Delta} - \bm{\wt{\pi}}^i) ))
   = \frac{d}{d \Delta} \left( -\frac{ \wt{\bm{\rho}}_{i,0}^\top D_{\tilde{G}_i}^{-1}\frac{d}{d \Delta} \wt{\bm{\rho}}_{i,\Delta}  }{\wt{\bm{\rho}}_{i,0}^\top D_{\tilde{G}_i}^{-1}(\wt{\bm{\rho}}_{i,\Delta}-\bm{\wt{\pi}}^i)} \right)\\
   &= \frac{d}{d \Delta} \left( \frac{ \wt{\bm{\rho}}_{i,0}^\top D_{\tilde{G}_i}^{-1} \mathcal{L}_i \wt{\bm{\rho}}_{i,\Delta} }{\wt{\bm{\rho}}_{i,0}^\top D_{\tilde{G}_i}^{-1}(\wt{\bm{\rho}}_{i,\Delta}-\bm{\wt{\pi}}^i)} \right)
   = \frac{   \wt{\bm{\rho}}_{i,0}^\top D_{\tilde{G}_i}^{-1}(\wt{\bm{\rho}}_{i,\Delta}-\bm{\wt{\pi}}^i) \wt{\bm{\rho}}_{i,0}^\top D_{\tilde{G}_i}^{-1} \mathcal{L}_i^2 \wt{\bm{\rho}}_{i,\Delta}
  - {(\wt{\bm{\rho}}_{i,0}^\top D_{\tilde{G}_i}^{-1} \mathcal{L}_i \wt{\bm{\rho}}_{i,\Delta} )}^2 }{{(\wt{\bm{\rho}}_{i,0}^\top D_{\tilde{G}_i}^{-1}(\wt{\bm{\rho}}_{i,\Delta}-\bm{\wt{\pi}}^i))}^2}.
  \end{align*}
  It is sufficient to check the positivity of the numerator.
  Note that the numerator can be written as
  \begin{align}
    (\wt{\bm{\rho}}_{i,0}^\top D_{\tilde{G}_i}^{-1}(\wt{\bm{\rho}}_{i,\Delta}-\bm{\wt{\pi}}^i))(\wt{\bm{\rho}}_{i,0}^\top D_{\tilde{G}_i}^{-1} \mathcal{L}_i^2 \wt{\bm{\rho}}_{i,\Delta} )
    - {(\wt{\bm{\rho}}_{i,0}^\top D_{\tilde{G}_i}^{-1} \mathcal{L}_i \wt{\bm{\rho}}_{i,\Delta} )}^2.
    \label{eq:log-second-derivative-1}
  \end{align}
%  Then, we vary the parts as forms of norms of vectors: The first part becomes
  The first factor of the first term of~\eqref{eq:log-second-derivative-1} is
  \[
    \wt{\bm{\rho}}_{i,0}^\top D_{\tilde{G}_i}^{-1}(\wt{\bm{\rho}}_{i,\Delta}-\bm{\wt{\pi}}^i) = \|D_{\tilde{G}_i}^{-1/2} (\wt{\bm{\rho}}_{i,\Delta/2} - \bm{\wt{\pi}}^i)\|^2
  \]
  by Proposition~\ref{prop:norm}.
  The second factor of the first term of~\eqref{eq:log-second-derivative-1} is
  \begin{align*}
    &\wt{\bm{\rho}}_{i,0}^\top D_{\tilde{G}_i}^{-1} \mathcal{L}_i^2 \wt{\bm{\rho}}_{i,\Delta}  = \wt{\bm{\rho}}_{i,0}^\top D_{\tilde{G}_i}^{-1} {(I - A_{\wt{G}_i} D_{\tilde{G}_i}^{-1})}^2 H_{i,\Delta} \wt{\bm{\rho}}_{i,0}\\
    %\wt{\bm{\rho}}_{i,0} H_{i,\Delta}(I - D_{\tilde{G}_i}^{-1}A_{\wt{G}_i})^2 D_{\tilde{G}_i}^{-1} \wt{\bm{\rho}}_{i,0}^\ast \\
    &= \wt{\bm{\rho}}_{i,0}^\top D_{\tilde{G}_i}^{-1} D_{\tilde{G}_i} {(H_{i,\Delta/2})}^\top D_{\tilde{G}_i}^{-1} {(I - A_{\wt{G}_i} D_{\tilde{G}_i}^{-1})}^2 H_{i,\Delta/2} \wt{\bm{\rho}}_{i,0} \\
    %\wt{\bm{\rho}}_{i,0} H_{i,\Delta/2}(I - D_{\tilde{G}_i}^{-1}A_{\wt{G}_i})^2 D_{\tilde{G}_i}^{-1} H^{i,\ast}_{\Delta/2} D_{\tilde{G}_i} D_{\tilde{G}_i}^{-1} \wt{\bm{\rho}}_{i,0}^\ast \\
    &= \wt{\bm{\rho}}_{i,0}^\top {(H_{i,\Delta/2})}^\top D_{\tilde{G}_i}^{-1} (D_{\tilde{G}_i} - A_{\wt{G}_i}) D_{\tilde{G}_i}^{-1} (D_{\tilde{G}_i}-A_{\wt{G}_i}) D_{\tilde{G}_i}^{-1} H_{i,\Delta/2} \wt{\bm{\rho}}_{i,0} \\
    % \wt{\bm{\rho}}_{i,0} H_{i,\Delta/2}D_{\tilde{G}_i}^{-1}(D_{\tilde{G}_i} - A_{\wt{G}_i}) D_{\tilde{G}_i}^{-1} (D_{\tilde{G}_i} - A_{\wt{G}_i}) D_{\tilde{G}_i}^{-1} H^{i,\ast}_{\Delta/2} \wt{\bm{\rho}}_{i,0}^\ast \\
    &= \|D_{\tilde{G}_i}^{-1/2} (D_{\tilde{G}_i}-A_{\wt{G}_i}) D_{\tilde{G}_i}^{-1} H_{i,\Delta/2} \wt{\bm{\rho}}_{i,0} \|^2\\
    &= \|D_{\tilde{G}_i}^{-1/2} (I-A_{\wt{G}_i} D_{\tilde{G}_i}^{-1}) \wt{\bm{\rho}}_{i,\Delta/2}\|^2\\
    &= \|D_{\tilde{G}_i}^{-1/2} \mathcal{L}_i \wt{\bm{\rho}}_{i,\Delta/2}\|^2.
  \end{align*}
  The second term of~\eqref{eq:log-second-derivative-1} is
  \begin{align}
    & \wt{\bm{\rho}}_{i,0}^\top D_{\tilde{G}_i}^{-1} \mathcal{L}_i \wt{\bm{\rho}}_{i,\Delta}
    = \wt{\bm{\rho}}_{i,0}^\top D_{\tilde{G}_i}^{-1} (I - A_{\wt{G}_i} D_{\tilde{G}_i}^{-1}) H_{i,\Delta} \wt{\bm{\rho}}_{i,0} \nonumber \\
    &= \wt{\bm{\rho}}_{i,0}^\top {(H_{i,\Delta/2})}^\top D_{\tilde{G}_i}^{-1} (I- A_{\wt{G}_i} D_{\tilde{G}_i}^{-1}) H_{i,\Delta/2} \wt{\bm{\rho}}_{i,0} \nonumber \\
    &= \wt{\bm{\rho}}_{i,\Delta/2}^\top D_{\tilde{G}_i}^{-1} (I-A_{\wt{G}_i} D_{\tilde{G}_i}^{-1}) \wt{\bm{\rho}}_{i,\Delta/2}
    = \wt{\bm{\rho}}_{i,\Delta/2}^\top D_{\tilde{G}_i}^{-1} \mathcal{L}_i \wt{\bm{\rho}}_{i,\Delta/2}.
    \label{eq:log-second-derivative-2}
  \end{align}
  We can rephrase~\eqref{eq:log-second-derivative-2} as the inner product of the vectors $D_{\tilde{G}_i}^{-1/2} \mathcal{L}_i\wt{\bm{\rho}}_{i,\Delta/2} $ and
  $D_{\tilde{G}_i}^{-1/2} (\wt{\bm{\rho}}_{i,\Delta/2} - \bm{\wt{\pi}}^i)$, as follows:
  \begin{align*}
    & {(D_{\tilde{G}_i}^{-1/2} \mathcal{L}_i\wt{\bm{\rho}}_{i,\Delta/2} )}^\top D_{\tilde{G}_i}^{-1/2} (\wt{\bm{\rho}}_{i,\Delta/2} - \bm{\wt{\pi}}^i)
     %\wt{\bm{\rho}}_{i,\Delta/2} \mathcal{L}_i D_{\tilde{G}_i}^{-1/2}((\wt{\bm{\rho}}_{i,\Delta/2} - \bm{\wt{\pi}}^i)D_{\tilde{G}_i}^{-1/2})^\ast
    = \wt{\bm{\rho}}_{i,\Delta/2}^\top \mathcal{L}_i^\top D_{\tilde{G}_i}^{-1/2}D_{\tilde{G}_i}^{-1/2} (\wt{\bm{\rho}}_{i,\Delta/2} - \bm{\wt{\pi}}^i)\\
    % \wt{\bm{\rho}}_{i,\Delta/2} \mathcal{L}_i D_{\tilde{G}_i}^{-1/2}D_{\tilde{G}_i}^{-1/2} (\wt{\bm{\rho}}^{i,\ast}_{\Delta/2} - \bm{\wt{\pi}}^i^\ast) \\
    &= \wt{\bm{\rho}}_{i,\Delta/2}^\top \mathcal{L}_i^\top D_{\tilde{G}_i}^{-1} \wt{\bm{\rho}}_{i,\Delta/2} -
     \wt{\bm{\rho}}_{i,\Delta/2}^\top \mathcal{L}_i^\top D_{\tilde{G}_i}^{-1}\bm{\wt{\pi}}^i
    = \wt{\bm{\rho}}_{i,\Delta/2}^\top \mathcal{L}_i^\top D_{\tilde{G}_i}^{-1} \wt{\bm{\rho}}_{i,\Delta/2},
  \end{align*}
  where the last equality follows from
  \[
   \wt{\bm{\rho}}_{i,\Delta/2}^\top \mathcal{L}_i^\top D_{\tilde{G}_i}^{-1}\bm{\wt{\pi}}^i = \wt{\bm{\rho}}_{i,\Delta/2}^\top \mathcal{L}_i^\top \frac{1}{\vol(V_i)} \1
  \]
  and $\mathcal{L}_i^\top\1 = D_{\tilde{G}_i}^{-1} (D_{\tilde{G}_i} - A_{\wt{G}_i})\1 = D_{\tilde{G}_i}^{-1} \0 = \0$.

  Hence,  we have
  \begin{align*}
    & \eqref{eq:log-second-derivative-1}
    =  \| D_{\tilde{G}_i}^{-1/2} \mathcal{L}_i \wt{\bm{\rho}}_{i,\Delta/2} \|^2 \cdot \|D_{\tilde{G}_i}^{-1/2} (\wt{\bm{\rho}}_{i,\Delta/2} - \bm{\wt{\pi}}^i)\|^2
    - {\left({(D_{\tilde{G}_i}^{-1/2} \mathcal{L}_i \wt{\bm{\rho}}_{i,\Delta/2})}^\top D_{\tilde{G}_i}^{-1/2} (\wt{\bm{\rho}}_{i,\Delta/2} - \bm{\wt{\pi}}^i)\right)}^2 \geq 0,
  \end{align*}
  where the last inequality follows from the Cauchy-Schwarz inequality.
\end{proof}

\subsection{Proof of Lemma~\ref{lem:equality-of-norms}}\label{subsec:equality-of-norms}
\begin{proof}% [Proof of Lemma~\ref{lem:equality-of-norms}]
 We recall that
\begin{align*}
 \wt{\bm{\rho}}_{i, \Delta/2}(u^i_k) =  \sum_{v \in U^i_k} \bm{\rho}_{i, \Delta/2}(v)
 = \left( \sum_{v \in U^i_k} d_G(v) \right) \bm{\mu}_{i, \Delta/2}(u^i_k) .
\end{align*}
Hence, we obtain
\begin{align*}
 \| \wt{\bm{\rho}}_{i, \Delta/2} - \wt{\bm{\pi}}^i \|_{D_{\tilde{G}_i}^{-1}}^2
 &= (\wt{\bm{\rho}}_{i, \Delta/2} - \wt{\bm{\pi}}^i)^\top
 D_{\tilde{G}_i}^{-1} (\wt{\bm{\rho}}_{i, \Delta/2} - \wt{\bm{\pi}}^i) \\
 & = \sum_{k = 1}^{m_i} \frac{1}{d_{\tilde{G}_i}(u^i_k)}(\wt{\bm{\rho}}_{i, \Delta/2}(u^i_k) - \wt{\bm{\pi}}^i(u^i_k))^2 \\
 & = \sum_{k = 1}^{m_i} \frac{1}{d_{\tilde{G}_i}(u^i_k)}\left(\sum_{u \in U^i_k}
 \bm{\rho}_{i, \Delta/2}(u) - \frac{d_{\tilde{G}_i}(u^i_k)}{\vol(\wt{V}_i)}\right)^2 \\
 &= \sum_{k = 1}^{m_i} d_{\tilde{G}_i}(u^i_k)
 \left( \bm{\mu}_{i, \Delta/2}(u^i_k) - \frac{1}{\vol(V)}\right)^2.
\end{align*}
On the other hand, the norm on $G$ becomes the following:
\begin{align*}
\| \bm{\rho}_{i, \Delta/2} - \bm{\pi}^i \|_{D^{-1}}^2
&= \sum_{u \in V} \frac{1}{d_G(u)}(\bm{\rho}_{i, \Delta/2}(u) - \bm{\pi}(u))^2 \\
&= \sum_{k = 1}^{m_i} \sum_{u \in U^i_k}\frac{1}{d_G(u)}\left(\bm{\rho}_{i, \Delta/2}(u) -
\frac{d_G(u)}{\vol(V)}\right)^2 \\
&= \sum_{k = 1}^{m_i} \sum_{u \in U^i_k}d_G(u)\left(\bm{\mu}_{i, \Delta/2}(u^i_k) -
\frac{1}{\vol(V)}\right)^2 \\
&= \sum_{k = 1}^{m_i} d_{\tilde{G}_i}(u_k)\left(\bm{\mu}_{i, \Delta/2}(u^i_k) -
\frac{1}{\vol(V)}\right)^2. \qedhere
\end{align*}

\end{proof}

\subsection{Proof of Lemma~\ref{lem:upper-bound}}\label{subsec:upper-bound}
We first derive a lower bound on the log derivative of $f_i(\Delta)$.
\begin{lem}\label{lem:lower-bound-on-log-derivative}
  For any $i \in \mathbb{Z}_+$ and $\Delta \geq 0$, we have
  \[
  -\frac{d}{d \Delta} \log f_i(\Delta)
%  = -\frac{ \wt{\bm{\rho}}_{i,0}^\top D_{\tilde{G}_i}^{-1}\frac{d}{d \Delta} \wt{\bm{\rho}}_{i,\Delta}  }{\wt{\bm{\rho}}_{i,0}^\top D_{\tilde{G}_i}^{-1}(\wt{\bm{\rho}}_{i,\Delta}-\wt{\bm{\pi}}^i)}
   \geq \frac{\wt{\kappa}_{i,\Delta/2}^2}{2}.
  \]
\end{lem}
\begin{proof}[Proof of Lemma~\ref{lem:lower-bound-on-log-derivative}]
  By Lemma~\ref{lem:log-derivative}, we have
  \[
    -\frac{d}{d \Delta} \log f_i(\Delta)
    = \mathcal{R}_i\left(\frac{\wt{\bm{\rho}}_{i,\Delta/2}}{d_{\wt{G}_i}} -\frac{1}{\vol(\wt{V}_i)}\right).
%    \frac{\sum_{uv \in \wt{E}_i} {\left( \frac{\wt{\bm{\rho}}_{i,\Delta/2}(u)}{d(u)} -\frac{\wt{\bm{\rho}}_{i,\Delta/2}(v)}{d(v)} \right)}^2 w_i(uv)}{\sum_{v \in V} {\left( \frac{\wt{\bm{\rho}}_{i,\Delta/2}(v)}{d(v)} -\frac{1}{\vol(V) }\right)}^2 d(v)}.
  \]
  Then, by applying Cheeger's inequality on the vector $\wt{\bm{\rho}}_{i,\Delta/2}/d_{\wt{G}_i}$, we obtain
  \[
    \max_{c \in \mathbb{R}}
    \mathcal{R}_i\left(\frac{\wt{\bm{\rho}}_{i,\Delta/2}}{d_{\wt{G}_i}} - c\right)
    %\frac{\sum_{uv \in \wt{E}_i} {\left( \frac{\wt{\bm{\rho}}_{i,\Delta/2}(u)}{d(u)} -\frac{\wt{\bm{\rho}}_{i,\Delta/2}(v)}{d(v)} \right)}^2 w_i(uv)}{\sum_{v \in V} {\left( \frac{\wt{\bm{\rho}}_{i,\Delta/2}(v)}{d(v)} -c\right)}^2 d(v)}
    \geq \frac{\wt{\kappa}_{i,\Delta/2}^2}{2}.
  \]
  Hence, it suffices to show that the left hand side (LHS) attains the maximum value when $c = 1/\vol(\wt{V}_i)$.
  Let $\varphi \colon \mathbb{R} \to \mathbb{R}$ be the denominator of the LHS (recall~\eqref{eq:Rayleigh-quotient}) as a function of $c$.
  Then,
  \[
    \varphi'(c) = - 2 \sum_{v \in \wt{V}_i} \left( \frac{\wt{\bm{\rho}}_{i,\Delta/2}(v)}{d_{\wt{G}_i}(v)} - c \right) d_{\wt{G}_i}(v).
  \]
  Hence $\varphi'(c) = 0$ yields
  \[
    \sum_{v \in \wt{V}_i} \wt{\bm{\rho}}_{i,\Delta/2}(v) - \left(\sum_{v \in \wt{V}_i} d_{\wt{G}_i}(v) \right) c = 0,
  \]
  which implies $c = 1/ \vol(\wt{V}_i)$ attains the minimum of $\varphi$.
  % This means that for any $x$,
  % we have
  % \[
  %  \varphi(c) \geq \sum_{w \in V} {\left( \frac{\wt{\bm{\rho}}_{i,\Delta/2}(w)}{d(w)} -\frac{1}{\vol(V) }\right)}^2 d(w).
  % \]
  % Let $u_1, \dots, u_n \in V$ be the ordered elements of $V$ satisfying
  % \[
  %  \frac{\wt{\bm{\rho}}_{i,\Delta}(u_1)}{d(u_1)} \geq \frac{\wt{\bm{\rho}}_{i,\Delta}(u_2)}{d(u_2)} \geq \cdots \geq \frac{\wt{\bm{\rho}}_{i,\Delta}(u_n)}{d(u_n)} \geq 0,
  % \]
  % and set $S_0 = \emptyset $, $S_j =\{ u_1, \dots, u_j \}$ ($j= 1, \dots, n$).
  % Let $r \in \{1, 2, \dots, n \}$ be the largest number satisfying $\vol(S_r) \leq \frac{\vol(V)}{2}$.
  % If we set
  % \[
  %  X = \frac{ \sum_{\{ u,v \}} {\left( \frac{\wt{\bm{\rho}}_{i,\Delta/2}(u)}{d(u)}
  %      -\frac{\wt{\bm{\rho}}_{i,\Delta/2}(v)}{d_v} \right)}^2 w_{u,v}^i}{\sum_w {\left( \frac{\wt{\bm{\rho}}_{i,\Delta/2}(w)}{d(w)} -\frac{\wt{\bm{\rho}}_{i,\Delta/2}(u_r)}{d_{u_r}}\right)}^2 d(w)},
  % \]
  % then by the minimality of $1/\vol(V)$ in the denominator,
  % \[
  %    - \frac{ \wt{\bm{\rho}}_{i,0}^\top D_{\tilde{G}_i}^{-1}\frac{d}{d \Delta} \wt{\bm{\rho}}_{i,\Delta}  }{\wt{\bm{\rho}}_{i,0}^\top D_{\tilde{G}_i}^{-1}(\wt{\bm{\rho}}_{i,\Delta}-\pi)}  \geq X.
  % \]
  % Applying the Cheeger's inequality (for graphs), we have $X \geq \frac{(\wt{\kappa}_{i,\Delta})^2}{2}$.
\end{proof}

\begin{proof}[Proof of Lemma~\ref{lem:upper-bound}]%\label{subsec:upper-bound}
%\begin{proof}[Proof of Lemma~\ref{lem:upper-bound}]
  We are now ready to prove Lemma~\ref{lem:upper-bound}. By Lemma~\ref{lem:lower-bound-on-log-derivative}, we have
  \begin{align*}
    \log f_{i_1}(2(t-t_{i_1})) -\log f_{i_1}(0) & \leq - \wt{\kappa}_{i_1,[0,\Delta]}^2 (t-t_{i_1}), \\
    \log f_j(2(t_{j+1}-t_j)) -\log f_j(0) & \leq - \wt{\kappa}_j^2(t_{j+1}-t_j) \quad (j = 0,\ldots,i-1), \\
    \log f_{i_0}(2(t_{i_0+1} - t_{i_0})) -\log f_{i_0}(2T - 2t_{i_0}) & \leq - \wt{\kappa}_{i_0,[T-t_{i_0}, t_{i_0+1}- t_{i_0}]}^2 (t_{i_0+1} - T),
  \end{align*}
  Hence, we have
  \begin{align*}
    & f_{i_1}(t - t_{i_1})
    \leq f_{i_1}(0) \exp\left({-\wt{\kappa}_{i_1,[0,t-t_{i_1}]}^2(t-t_{i_1})}\right)
    = \|\wt{\bm{\rho}}_{i_1,0} - \wt{\bm{\pi}}^{i_1}\|_{D_{\tilde{G}_{i_1}}^{-1}}^2 \exp\left({-\wt{\kappa}_{i_1,[0,t-t_{i_1}]}^2(t-t_{i_1})}\right) \\
    & = \|\wt{\bm{\rho}}_{i_1-1,t_{i_1} - t_{(i_1-1)}} - \wt{\bm{\pi}}^{(i_1-1)}\|_{D_{\wt{G}_{(i_1-1)}}^{-1}}^2 \exp\left({-\wt{\kappa}_{i_1,[0,t-t_{i_1}]}^2(t-t_{i_1})}\right) \\
    &= f_{(i_1-1)}(2(t_{i_1}-t_{(i_1-1)})) \exp\left({-\wt{\kappa}_{i_1,[0,t-t_{i_1}]}^2(t-t_{i_1})}\right) \\
    & \leq f_{(i_1-1)}(0)\left( -\wt{\kappa}_{i_1,[0,t-t_{i_1}]}^2(t-t_{i_1}) - \wt{\kappa}_{(i_1-1)}^2(t_{i_1}-t_{(i_1-1)}) \right) \leq \cdots \\
    & \leq f_{i_0}(2T-2t_{i_0}) \exp\left(-\wt{\kappa}_{i_1,[0,t-t_{i_1}]}^2(t-t_{i_1}) - \sum_{j=i_0+1}^{i_1-1}\wt{\kappa}_{j}^2(t_{j+1}-t_{j}) - \wt{\kappa}_{i_0,[T-t_{i_0},t_{i_0+1}-t_{i_0}]}^2(t_{i_0+1}-T) \right)\\
    &\leq \| \bm{\rho}_T^{\bm{\pi}_v} - \bm{\pi} \|_{ D^{-1}}^2 \exp ({-\wt{\kappa}_{T,t}^v}^2(t-T)). \qedhere
  \end{align*}
\end{proof}

\subsection{Proof of Lemma~\ref{lem:lower-bound}}\label{subsec:lower-bound}
\begin{proof}%[Proof of Lemma~\ref{lem:lower-bound}]%\label{subsec:lower-bound}
As in \cite[Lemma 4.11, 3]{chan2018spectral}, 
the derivative of the Rayleigh quotient 
$
 \langle \bm{\rho}_t^{\bm{\pi}_v}, \ca{L} 
 \bm{\rho}_t^{\bm{\pi}_v} \rangle_{D^{-1}}/ 
 \| \bm{\rho}_{t}^{\bm{\pi}_v} - \bm{\pi}\|^2_{D^{-1}}
$
is non-positive, hence this does not increase about $t$. By this monotonicity, we have 
\[
 - \frac{d}{dt} \log \| \bm{\rho}_t^{\bm{\pi}_v} - \bm{\pi} \|_{D^{-1}}^2  = 2\frac{\langle 
 \bm{\rho}_t^{\bm{\pi}_v}, \ca{L}\bm{\rho}_t^{\bm{\pi}_v}\rangle_{D^{-1}}}
 {\| \bm{\rho}_t^{\bm{\pi}_v} - \bm{\pi} \|^2_{D^{-1}}} 
 \leq 2\frac{\langle 
 \bm{\rho}_T^{\bm{\pi}_v}, \ca{L}\bm{\rho}_T^{\bm{\pi}_v}\rangle_{D^{-1}}}
 {\| \bm{\rho}_T^{\bm{\pi}_v} - \bm{\pi} \|^2_{D^{-1}}} = g_v(T).
\]
By integrating this on $[T, t]$, we obtain the claimed inequality. 
\end{proof}

\subsection{Proof of Lemma~\ref{lem:sweep}}
\begin{proof}%[Proof of Lemma~\ref{lem:sweep}]
Let $\{ U_1, U_2,  \dots, U_m\} \subseteq 2^V$ be the $\sigma^\ast$-equivalence classes such that $U_k \succ U_{k+1}$ ($k=1, \dots, m-1$), i.e., for any
$u \in U_k, v \in U_{k+1}$, $\bm{x}(u) > \bm{x}(v)$.
Then, the sweep set $S$ can be written by
\[
 S^a = S_i := U_1 \cup \cdots \cup U_i
\]
for a certain integer $i$. We recall that the conductance of this $S$ on $G$ is
\[
 \phi_G(S_i) = \frac{ \sum_{e \in E, e\cap S_i \neq \emptyset \atop
 e\cap V\!\setminus \!S_i \neq \emptyset} w_e }{\min \{ \vol(S_i), \vol(V\!\setminus\! S_i) \}}.
\]
Now, $\wt{S}^a$ is equal to $\wt{S}_i = \{ u_1, u_2, \dots, u_i\} $ for same $i$. Then,
the conductance $\phi_{\wt{G}}(\wt{S}_i)$ is
\[
 \phi_{\wt{G}}(\wt{S}_i) = \frac{\sum_{uv \in \tilde{E}, uv\cap \tilde{S}_i \neq \emptyset \atop
 uv\cap \tilde{V}\!\setminus \!\tilde{S}_i \neq \emptyset} \wt{w}(uv) }
 {\min \{ \vol(\wt{S}_i), \vol(\wt{V}\!\setminus\! \wt{S}_i) \}}.
\]
By simple calculation, we can show that the denominators are equal.
We check the equality of the numerators here.
\begin{align*}
 \sum_{uv \in \tilde{E}, uv\cap \tilde{S}_i \neq \emptyset \atop
 uv\cap \tilde{V}\!\setminus \!\tilde{S}_i \neq \emptyset} \wt{w}(uv)
 &= \sum_{j\leq i} \sum_{k \geq i+1} \wt{w}(u_ju_k) \\
 &= \sum_{j\leq i} \sum_{k \geq i+1} \left( \sum_{e\in E, S_e^{\sigma^\ast}\cap U_j \neq \emptyset \atop I_e^{\sigma^\ast}\cap U_k \neq \emptyset} w_e +
 \sum_{e\in E, S_e^{\sigma^\ast}\cap U_j \neq \emptyset \atop I_e^{\sigma^\ast}\cap U_k \neq \emptyset} w_e
  \right) \\
  &= \sum_{e\in E, S_e^{\sigma^\ast}\cap S_i \neq \emptyset \atop I_e^{\sigma^\ast}\cap
  V\!\setminus\! S_i \neq \emptyset} w_e
  = \sum_{e\in E, e\cap S_i \neq \emptyset \atop e\cap
  V\!\setminus\! S_i \neq \emptyset} w_e.
\end{align*}
Hence, the numerators are also the same.
%  We have
%  \begin{align*}
%     & \sum_{uv \in \partial_{G'}(S)} w'(uv) = \sum_{uv \in \partial_{G'}(S)} \sum_{e \in E} w'_e(uv)
%     =  \sum_{e \in E} \sum_{uv \in \partial_{G'}(S)} w'_e(uv)
%     =  \sum_{e \in \partial_G(S)} \sum_{uv \in \partial_{G'}(S)} w'_e(uv) \\
%     & \leq \sum_{e \in \partial_G(S)} \sum_{u \in S_e, v \in I_e} w'_e(uv)
%     = \sum_{e \in \partial_G(S)} w(e).
%  \end{align*}
%  Thus, we have $\phi_{G'}(S) \leq \phi_G(S)$ as $\vol_{G'}(S)=\vol_{G}(S)$.
%  In addition, the equality holds when $S$ is a sweep set with respect to $\bm{x}$, because $\sum_{uv \in \partial_{G'}(S)} w'_e(uv) = \sum_{u \in S_e, v \in I_e} w'_e(uv)$ holds for every hyperedge $e \in E$.
\end{proof}


\begin{thebibliography}{10}

\bibitem{Agarwal:2006ci}
Sameer Agarwal, Kristin Branson, and Serge Belongie.
\newblock Higher order learning with graphs.
\newblock In {\em Proceedings of the 23rd International Conference on Machine
  Learning (ICML)}, pages 17--24, 2006.

\bibitem{Alon:1986gz}
Noga Alon.
\newblock Eigenvalues and expanders.
\newblock {\em Combinatorica}, 6(2):83--96, 1986.

\bibitem{Alon:1985jg}
Noga Alon and V~D Milman.
\newblock $\lambda_1$, isoperimetric inequalities for graphs, and
  superconcentrators.
\newblock {\em Journal of Combinatorial Theory, Series B}, 38(1):73--88, 1985.

\bibitem{Chan:2018eu}
T-H~Hubert Chan, Anand Louis, Zhihao~Gavin Tang, and Chenzi Zhang.
\newblock {Spectral properties of hypergraph Laplacian and approximation
  algorithms}.
\newblock {\em Journal of the ACM}, 65(3):15--48, 2018.

\bibitem{chan2018spectral}
T-H~Hubert Chan, Anand Louis, Zhihao~Gavin Tang, and Chenzi Zhang.
\newblock Spectral properties of hypergraph laplacian and approximation
  algorithms.
\newblock {\em Journal of the ACM (JACM)}, 65(3):15, 2018.

\bibitem{chan2017diffusion}
T-H~Hubert Chan, Zhihao~Gavin Tang, Xiaowei Wu, and Chenzi Zhang.
\newblock Diffusion operator and spectral analysis for directed hypergraph
  laplacian.
\newblock {\em arXiv preprint arXiv:1711.01560}, 2017.

\bibitem{Chung:2007ep}
Fan Chung.
\newblock The heat kernel as the pagerank of a graph.
\newblock {\em Proceedings of the National Academy of Sciences of the United
  States of America}, 104(50):19735--19740, 2007.

\bibitem{fujii2018polynomial}
Kaito Fujii, Tasuku Soma, and Yuichi Yoshida.
\newblock Polynomial-time algorithms for submodular {L}aplacian systems.
\newblock {\em arXiv preprint, arXiv:1803.10923}, 2018.

\bibitem{Fujishige:2005uc}
Satoru Fujishige.
\newblock {\em Submodular Functions and Optimization}.
\newblock Elsevier, 2005.

\bibitem{Hein:2013wc}
Matthias Hein, Simon Setzer, Leonardo Jost, and Syama~Sundar Rangapuram.
\newblock The total variation on hypergraphs --- {L}earning on hypergraphs
  revisited.
\newblock In {\em Proceedings of the 27th Annual Conference on Neural
  Information Processing Systems (NIPS)}, pages 2427--2435, 2013.

\bibitem{Kloster:2014wq}
Kyle Kloster and David~F Gleich.
\newblock Heat kernel based community detection.
\newblock In {\em Proceedings of the 20th ACM SIGKDD International Conference
  on Knowledge Discovery and Data Mining (KDD)}, pages 1386--1395, 2014.

\bibitem{komura1967nonlinear}
Yukio Komura.
\newblock Nonlinear semi-groups in hilbert space.
\newblock {\em Journal of the Mathematical Society of Japan}, 19(4):493--507,
  1967.

\bibitem{Li:2018we}
Pan Li and Olgica Milenkovic.
\newblock Submodular hypergraphs: p-{L}aplacians, {C}heeger inequalities and
  spectral clustering.
\newblock In {\em Proceedings of the 35th International Conference on Machine
  Learning (ICML)}, pages 3014--3023, 2018.

\bibitem{miyadera1992nonlinear}
Isao Miyadera.
\newblock {\em Nonlinear Semigroups}, volume 109.
\newblock American Mathematical Soc., 1992.

\bibitem{Raghavendra:2010jg}
Prasad Raghavendra and David Steurer.
\newblock Graph expansion and the unique games conjecture.
\newblock In {\em Proceedings of the 42nd ACM Annual Symposium on Theory of
  Computing (STOC)}, pages 755--764, 2010.

\bibitem{Raghavendra:2012fc}
Prasad Raghavendra and David Steurer.
\newblock Reductions between expansion problems.
\newblock In {\em Proceedings of the 27th IEEE Annual Conference on
  Computational Complexity (CCC)}, pages 64--73, 2012.

\bibitem{Scholkopf:2006vj}
Bernhard Sch{\"o}lkopf, John Platt, and Thomas Hofmann.
\newblock Learning with hypergraphs: Clustering, classification, and embedding.
\newblock In {\em Proceedings of the 19th Annual Conference on Neural
  Information Processing Systems (NIPS)}, pages 1601--1608, 2006.

\bibitem{showalter2013monotone}
Ralph~Edwin Showalter.
\newblock {\em Monotone Operators in Banach Space and Nonlinear Partial
  Differential Equations}, volume~49.
\newblock American Mathematical Soc., 2013.

\bibitem{Yoshida:2016ig}
Yuichi Yoshida.
\newblock Nonlinear {L}aplacian for digraphs and its applications to network
  analysis.
\newblock In {\em Proceedings of the 9th ACM International Conference on Web
  Search and Data Mining (WSDM)}, pages 483--492, 2016.

\bibitem{Yoshida:2017zz}
Yuichi Yoshida.
\newblock Cheeger inequalities for submodular transformations.
\newblock In {\em Proceedings of the Thirtieth Annual ACM-SIAM Symposium on
  Discrete Algorithms (SODA)}, pages 2582--2601. SIAM, 2019.

\bibitem{Zhang:2017va}
Chenzi Zhang, Shuguang Hu, Zhihao~Gavin Tang, and T-H~Hubert Chan.
\newblock Re-revisiting learning on hypergraphs: Confidence interval and
  subgradient method.
\newblock In {\em Proceedings of the 34th International Conference on Machine
  Learning (ICML)}, pages 4026--4034, 2017.

\end{thebibliography}
\end{document}